\documentclass[10pt, a4paper]{article}%
\usepackage{amsfonts}
\usepackage{amssymb}
\usepackage{amsmath}
\usepackage{graphicx}
\usepackage{tikz}
\usepackage{pgflibraryarrows}
\usepackage{pgflibrarysnakes}%
\providecommand{\U}[1]{\protect\rule{.1in}{.1in}}
\usetikzlibrary{patterns}
\providecommand{\U}[1]{\protect\rule{.1in}{.1in}}
\newtheorem{theorem}{Theorem}

\newtheorem{acknowledgement}[theorem]{Acknowledgement}

\newtheorem{corollary}[theorem]{Corollary}

\newtheorem{definition}[theorem]{Definition}

\newtheorem{lemma}[theorem]{Lemma}

\newtheorem{proposition}[theorem]{Proposition}
\newtheorem{remark}[theorem]{Remark}

\newenvironment{proof}[1][Proof]{\noindent\textbf{#1.} }{\ \rule{0.5em}{0.5em}}
\begin{document}

\date{}
\title{On the blow up and condensation of supercritical solutions of the Nordheim equation for bosons.}
\maketitle

\bigskip

\begin{center}
\bigskip M. Escobedo\footnotemark[1]$^{,}$\footnotemark[2], J. J. L.
Vel\'{a}zquez\footnotemark[3] \bigskip
\end{center}

\footnotetext[1]{Departamento de Matem\'{a}ticas. Universidad del Pa\'{\i}s
Vasco UPV/EHU. Apartado 644. E-48080 Bilbao, Spain. E-mail:
miguel.escobedo@ehu.es} \footnotetext[2]{Basque Center for Applied Mathematics
(BCAM), Alameda de Mazarredo 14, E--48009 Bilbao, Spain.} \footnotetext[3]%
{Institute of Applied Mathematics, University of Bonn, Endenicher Allee 60,
53115 Bonn, Germany. E-mail: velazquez@iam.uni-bonn.de}

\noindent\textbf{Abstract.} In this paper we prove that the solutions of the isotropic, spatially homogeneous Nordheim equation for bosons, with bounded initial, data blow up in finite time in the $L^\infty$ norm if the values of the energy and particle density are in the range of values where the corresponding equilibria contains a Dirac mass.
We also prove that,  in  the weak solutions, whose initial data are measures with values of particle and energy densities satisfying the previous condition, a Dirac measure at the origin forms in finite time.


\section{Introduction.}

\setcounter{equation}{0} \setcounter{theorem}{0} In this paper we continue the
study of the blow-up properties of the solutions of the homogeneous Nordheim
equation initiated in \cite{EV1}. This system describes the dynamics of a
dilute homogeneous quantum gas of bosons. We will denote as $F\left(
t,p\right)  ~$the distribution of particles in the momentum space. The
evolution of $F$ is given by the following system of equations (cf.
\cite{Nor}):%
\begin{align}
\partial_{t}F_{1} &  =\int_{\mathbb{R}^{3}}\int_{\mathbb{R}^{3}}%
\int_{\mathbb{R}^{3}}q\left(  F\right)  \mathcal{M}d^{3}p_{2}d^{3}p_{3}%
d^{3}p_{4}\ \ ,\ \ p_{1}\in\mathbb{R}^{3}\ \ ,\ \ t>0\label{E2}\\
F_{1}\left(  0,p\right)   &  =F_{0}\left(  p\right)  \ \ ,\ \ p_{1}%
\in\mathbb{R}^{3}\label{E2b}%
\end{align}

\begin{align}
q\left(  F\right)   &  =q_{3}\left(  F\right)  +q_{2}\left(  F\right)
\ \ \ \ ,\ \ \ \ \ \epsilon=\frac{\left\vert p\right\vert ^{2}}{2}%
\label{T4E1a}\\
\mathcal{M} &  =\mathcal{M}\left(  p_{1},p_{2};p_{3},p_{4}\right)
=\delta\left(  p_{1}+p_{2}-p_{3}-p_{4}\right)  \delta\left(  \epsilon
_{1}+\epsilon_{2}-\epsilon_{3}-\epsilon_{4}\right) \label{T4E1b}%
\end{align}%
\begin{align}
q_{3}\left(  F\right)   &  =F_{3}F_{4}\left(  F_{1}+F_{2}\right)  -F_{1}%
F_{2}\left(  F_{3}+F_{4}\right) \label{Q1E1}\\
q_{2}\left(  F\right)   &  =F_{3}F_{4}-F_{1}F_{2}\label{Q1E2}%
\end{align}
where we use the notation $F_{j}=F\left(  t,p_{j}\right)  ,\ j\in
\mathbb{R}^{3}.$ Notice that, since we consider homogeneous distributions, the
density $F$ measures the number of particles per unit of volume, i.e., the
number of particles with moment in the cube $\left[  p,p+d^{3}p\right]  $ in a
volume $V$ would be given by $F\left(  p\right)  Vd^{3}p.$

In the case of isotropic distributions the system (\ref{E2})-(\ref{Q1E2}) can
be rewritten in a simpler form. The isotropy of the solutions means that
$F\left(  t,p\right)  =F\left(  t,\mathcal{R}p\right)  \ \ $for
any$\ \ \mathcal{R}\in SO\left(  3\right)  \ ,\ \ p\in\mathbb{R}^{3}%
,\ \ t\geq0.$ Then, there exists a function $f=f\left(  \epsilon,t\right)  $
where $\epsilon$ is as in (\ref{T4E1a}) such that $f\left(  t,\epsilon\right)
=F\left(  t,p\right)  \ $which solves:%
\begin{equation}
\partial_{t}f_{1}=\frac{8\pi^{2}}{\sqrt{2}}\int_{0}^{\infty}\int_{0}^{\infty
}q\left(  f\right)  Wd\epsilon_{3}d\epsilon_{4}\ \label{F3E2}%
\end{equation}
where:%
\begin{equation}
W=\frac{\min\left\{  \sqrt{\epsilon_{1}},\sqrt{\epsilon_{2}},\sqrt
{\epsilon_{3}},\sqrt{\epsilon_{4}}\right\}  }{\sqrt{\epsilon_{1}}%
}\ \ ,\ \ \epsilon_{2}=\epsilon_{3}+\epsilon_{4}-\epsilon_{1}\ \label{F3E3}%
\end{equation}
and $q\left(  \cdot\right)  $ is as in (\ref{T4E1a}) with $\epsilon
_{2}=\epsilon_{3}+\epsilon_{4}-\epsilon_{1}$. More details about this
computation as well as additional information about the physics of the
Nordheim equation can be found in \cite{EV1}.

A theory of global weak solutions for (\ref{F3E2}), (\ref{F3E3})
has been developed by X. Lu in the papers \cite{Lu1, Lu2, Lu3}. Uniqueness of the weak solutions for a given
initial datum is unknown. On the other hand,  the local (in time) existence and uniqueness of mild solutions of (\ref{F3E2}), (\ref{F3E3}) was proved in \cite{EV1}. These results are recalled in detail in Section 2.
\\
A different concept of solution for (\ref{F3E2}) has
been introduced in \cite{EMV1}, \cite{EMV2} where a class of classical
solutions of (\ref{F3E2}) which behave as \thinspace$f\left(  \epsilon
,t\right)  \sim a\left(  t\right)  \epsilon^{-\frac{7}{6}}$ as $\epsilon
\rightarrow0$ has been obtained. Such solutions have a decreasing number of
particles with $\epsilon>0$.

One of the main interests of the study of (\ref{E2})-(\ref{Q1E2}) or its
isotropic counterpart (\ref{F3E2}) is the relation of these systems with the
dynamical formation of Bose-Einstein condensates (cf. \cite{JPR}, \cite{LLPR},
\cite{ST1}, \cite{ST2}). It can be argued, on physical grounds, that the
steady states of (\ref{E2}), (\ref{T4E1a})-(\ref{Q1E2}) are the Bose-Einstein
distributions:%
\begin{equation}
F_{BE}\left(  p\right)  =m_{0}\delta\left(  p-p_{0}\right)  +\frac{1}%
{\exp\left(  \frac{\beta\left\vert p-p_{0}\right\vert ^{2}}{2}+\alpha\right)
-1}\ \label{St1}%
\end{equation}
where $m_{0}\geq0,$ $\beta\in\left(  0,\infty\right]$,  $0\leq\alpha<\infty$
and $\alpha\, m_{0}=0,\ p_{0}\in\mathbb{R}^{3}$.\ The precise sense in
which the distributions $F_{BE}$ are steady states requires further
clarification, because the right-hand side of (\ref{E2}) is not defined in
general if $F$ contains Dirac measures. The main reason to consider the
distributions $F_{BE}$ as the steady states of (\ref{E2}), (\ref{T4E1a}%
)-(\ref{Q1E2}) is that these distributions maximize the entropy of the system
of bosons under consideration for a given value of the momentum, energy and
number of particles of the system for unit of volume (cf. (\cite{Huang})). 
It has been  proved in \cite{Lu2} that the weak solutions of  (\ref{F3E2}), (\ref{F3E3}) 
 converge, as $t\to \infty$, to the only equilibrium with the same particle number and energy then their initial data.

One of the most peculiar features of the steady states $F_{BE}$ in (\ref{St1})
is the possible presence of a macroscopic fraction of particles at the value
$p=p_{0}.$ This feature is known as Bose-Einstein condensation. We will assume
in the following that the total momentum of the system is zero. This can be
always assumed choosing a suitable inertial system. In this case the
thermodynamics of the system of bosons can be described by means of two
quantities, namely the number of particles by unit of volume $M$ and the
energy for unit of volume $E:$%
\[
M=\int F_{BE}\left(  p\right)  d^{3}p\ \ ,\ \ E=\int F_{BE}\left(  p\right)
\frac{\left\vert p\right\vert ^{2}}{2}d^{3}p
\]

It is possible to associate unique values of $m_{0},\ \beta$ and $\alpha$,
with $\alpha\cdot m_{0}=0$ to any pair of values $M>0,\ E>0.$ We can split the
set of all the values for the particle density and energy $\left\{  \left(
M,E\right)  :M>0,\ E>0\right\}  $ in two different phases, namely the set of
values for which the corresponding value of $m_{0}$ in (\ref{St1}) is
positive, and those values of $M, E$ for which $m_{0}=0.$ We will say that
the first class of states has a Bose-Einstein condensate, while the second
class of states does not have a condensate. It is not hard to see (cf.
Subsection \ref{StatSol}) that the set of values $\left(  M, E\right)  $ for
which there is a nontrivial condensate is characterized by:%
\begin{equation}
M>\frac{\zeta\left(  \frac{3}{2}\right)  }{\left(  \zeta\left(  \frac{5}%
{2}\right)  \right)  ^{\frac{3}{5}}}\left(  \frac{4\pi}{3}\right)  ^{\frac
{3}{5}}E^{\frac{3}{5}}\label{S1E1}%
\end{equation}
where $\zeta\left(  \cdot\right)  $ is the Riemann zeta function.

The theory of Bose-Einstein condensates described above is meaningful for
particle systems described by the equilibrium distributions $F_{BE}$ in
(\ref{St1}). On the other hand, the equation (\ref{E2})-(\ref{Q1E2}) can be
used to describe the evolution of a large class of particle distributions $F_{0}$
satisfying some general condtions which ensure that the particle and energy
densities remain constant for arbitrary times. It has been suggested in the
physical literature (cf. \cite{JPR}, \cite{LLPR}, \cite{ST1}, \cite{ST2}) that
the onset of a macroscopic fraction of particles at the lowest energy level
(i.e. a condensate) is related to the blow-up at some finite time $T^*$ of the solutions of (\ref{E2}),
(\ref{T4E1a})-(\ref{Q1E2}). More precisely, the numerical simulations in
\cite{LLPR}, \cite{ST1}) show the existence of isotropic solutions of
(\ref{F3E2}), with initial data $f\left(  0,\epsilon\right)  =f_{0}\left(
\epsilon\right)  $, that become unbounded in finite time. It has actually been proved in \cite{EV1} that there exist a large
class of bounded initial distributions $f_{0}$ for which the corresponding
solutions of (\ref{F3E2}) blow-up in finite time. To wit, a blow-up condition
has been obtained which states that for any initial particle distribution with
a large concentration of particles with small energy, the corresponding
solution of (\ref{F3E2}) becomes unbounded in finite time. 

The scenario for the dynamical formation of Bose-Einstein condensates presented in   (cf. \cite{JPR}, \cite{LLPR}, \cite{ST1}, \cite{ST2}) suggests that, after this finite time blow up,  the measure  $g\left(  t,\epsilon\right)  =4\pi
\sqrt{2\epsilon}f\left(  t,\epsilon\right)  $  contains a Dirac mass at the origin for all later times $t>T^*$. It was proved in \cite{EV1} that if the initial data  $g_0$ satisfies a certain condition, that amounts to have a sufficiently large concentration of particles with small energy,  then  $g(t, \epsilon)$ contains a Dirac mass at the origin for all times greater than some finite $T _{ cond }\ge T^*$.  It has not been proved that $T _{ cond }=T^*$.  

We now define precisely the concept of subcritical and supercritical data. Given a
distribution of particles $f_{0}$ we can evaluate the corresponding number of
particles and energy respectively by means of:%
\[
4\pi\int_{0}^{\infty}f_{0}\left(  \epsilon\right)  \sqrt{2\epsilon}%
d\epsilon=M\ \ ,\ \ 4\pi\int_{0}^{\infty}f_{0}\left(  \epsilon\right)
\sqrt{2\epsilon^{3}}d\epsilon=E
\]

We will say that $f_{0}$ is supercritical if the inequality (\ref{S1E1})
holds. If, on the contrary, we have $M<\frac{\zeta\left(  \frac{3}{2}\right)
}{\left(  \zeta\left(  \frac{5}{2}\right)  \right)  ^{\frac{3}{5}}}\left(
\frac{4\pi}{3}\right)  ^{\frac{3}{5}}E^{\frac{3}{5}}$ we will say that the
distribution $f_{0}$ is subcritical. Notice that given a supercritical
distribution, the corresponding stationary distribution having the same amount
of particles and energy would have a nontrivial condensate.

We prove in this paper that any solution of Nordheim equation, initially
bounded and supercritical,  becomes unbounded in finite time $T _{ max}$ and develops a condensate in   finite time $T_0\ge T _{ max }$.

An interesting fact that is worth noticing is that the quadratic terms in
(\ref{E2}), (cf. (\ref{Q1E2})) that are in principle lower order terms, play a
crucial role in the onset of blow-up. The reason is that they produce a
transfer of mass from energies of order one to small energies.

We now sketch the main idea in the proof of the blow up  result of this paper. 
The
key ingredient is the blow-up condition for mild solutions in \cite{EV1}. There are two main
conditions needed to apply such blow-up condition. First we need to have
$f\geq\nu>0$ in some average sense for small values of $\epsilon.$ On the
other hand, we need to have an amount of mass of order $\rho^{\theta_{*}}$,
for some $\theta_{*}>0$, in a interval $\left[  0,\rho\right]  ,$ with $\rho$
small. In this paper we prove that every bounded solution of the Nordheim
equation satisfies these conditions for sufficiently long times. To this end
we use mainly two arguments. We first prove that the quadratic terms of
the\ Nordheim equation transport an amount of mass towards small values of the
energy, and that this transport cannot be balanced by the corresponding loss
terms. As a consequence we obtain an amount of mass of order $\nu R^{\frac
{3}{2}}$ for some $\nu>0$ in any interval $\left[  0,R\right]  $ for
sufficiently long times. The second argument that we use is the existence of
an increasing entropy for the Nordheim equation. The corresponding formula for
the dissipation of the entropy can be used to prove that, at least along some
sequences of time, the distribution $f$ approaches one stationary state of the
Nordheim equation. The total energy and number of particles of the admissible
stationary states obtained corresponds to the supercritical regime of the
Bose-Einstein condensation.\ This implies that there exists a positive amount
of mass in a small interval $\left[  0,\rho\right]  $ with $\rho>0$ small for
times sufficiently long. The condition in \cite{EV1} implies then blow-up in
finite time for the solutions of the Nordheim equation.

Similar arguments are applied to the weak solutions of  (\ref{F3E2}), (\ref{F3E3}) in order to prove
condensation in finite time.

We also remark that the blow-up  and the condensation conditions are local. Therefore, it is not
difficult to obtain initial data $f_{0}\left(  x\right)  $ yielding blow-up or condensation in
finite time for subcritical initial data.

The plan of this paper is the following. In Section 2 we recall the concept of
solution which we will use as well as the Local Well-Posedness Theorem proved
in \cite{EV1}. Section 3 states the main  results proved in this paper.
Section 4 summarizes some results proved in \cite{EV1} which will be used in
the proof of the main result of this paper. Section 5 contains one of the key
ingredients of this paper, namely, a lower estimate for the amount of mass
contained in regions $\epsilon\in\left[  0,R\right]  $ with $R$ small. This
estimate is due to the effect of the quadratic, Boltzmann terms in
(\ref{F3E2}). Section 6 contains a rigorous proof of the entropy dissipation
formula for the solutions considered in this paper. Section 7 contains a
technical Lemma that estimates the maximum amount of particles contained in
regions $\epsilon\geq R,$ with $R$ small, in terms of the total energy of the
solution. Section 8 proves that for the so-called supercritical data, the
distribution of particles $g\left(  t,\epsilon\right)  $ contains a positive
amount of mass in the region where $\epsilon$ is small, at least for some
subsequences of time. Section 9 contains the proof of the main blow up theorem of this
paper. In Section 10 we prove  the formation in finite time of  a Dirac mass at the origin for 
weak solutions whose mass and energy are supercritical.
Section 11 explains how to apply the result in \cite{EV1} in order to
obtain finite time blow-up and condensation  for the solutions of (\ref{S1E1}) for some subcritical
distributions of particles. 
\section{Well-posedness results: Weak and mild solutions}

\setcounter{equation}{0} \setcounter{theorem}{0}

In order to formulate the main  results of this paper we need to
recall some previous results about equation (\ref{F3E2}), (\ref{F3E3}).

We first recall a well-posedness result obtained in \cite{EV1}.To this end we precise
the concept of mild solution of (\ref{F3E2}), (\ref{F3E3}) that we will use.

\subsection{Definition of mild solutions.}

Given $\gamma\in\mathbb{R}$ we will denote as $L^{\infty}\left(
\mathbb{R}^{+};\left(  1+\epsilon\right)  ^{\gamma}\right)  $ the space of
functions such that:%
\[
\left\Vert f\right\Vert _{L^{\infty}\left(  \mathbb{R}^{+};\left(
1+\epsilon\right)  ^{\gamma}\right)  }=\sup_{\epsilon\geq0}\left\{  \left(
1+\epsilon\right)  ^{\gamma}f\left(  \epsilon\right)  \right\}  <\infty
\]
Notice that $L^{\infty}\left(  \mathbb{R}^{+};\left(  1+\epsilon\right)
^{\gamma}\right)  $ is a Banach space with the norm $\left\Vert \cdot
\right\Vert _{L^{\infty}\left(  \mathbb{R}^{+};\left(  1+\epsilon\right)
^{\gamma}\right)  }.$ Given $T_{2}>T_{1}>0,$ we define $L_{loc}^{\infty
}\left(  \left[  T_{1},T_{2}\right)  ;L^{\infty}\left(  \mathbb{R}^{+};\left(
1+\epsilon\right)  ^{\gamma}\right)  \right)  $ as the set of functions
satisfying:%
\[
\sup_{t\in K}\left\Vert f\left(  t,\cdot\right)  \right\Vert _{L^{\infty
}\left(  \mathbb{R}^{+};\left(  1+\epsilon\right)  ^{\gamma}\right)  }<\infty
\]
for any compact $K\subset\left[  T_{1},T_{2}\right)  .$ Notice that these
spaces are not Banach spaces. We also define the space $L^{\infty}\left(
\left[  T_{1},T_{2}\right]  ;L^{\infty}\left(  \mathbb{R}^{+};\left(
1+\epsilon\right)  ^{\gamma}\right)  \right)  $ which is the Banach space of
functions such that:
\[
\left\Vert f\right\Vert _{L^{\infty}\left(  \left[  T_{1},T_{2}\right]
;L^{\infty}\left(  \mathbb{R}^{+};\left(  1+\epsilon\right)  ^{\gamma}\right)
\right)  }=\sup_{t\in\left[  T_{1},T_{2}\right]  }\left\Vert f\left(
t,\cdot\right)  \right\Vert _{L^{\infty}\left(  \mathbb{R}^{+};\left(
1+\epsilon\right)  ^{\gamma}\right)  }<\infty
\]

\begin{definition}
\label{mild}Suppose that $\gamma>3,\ T_{2}>T_{1}>0.$ We will say that a
function
\[
f\in L_{loc}^{\infty}\left(  \left[  T_{1},T_{2}\right)  ;L^{\infty}\left(
\mathbb{R}^{+};\left(  1+\epsilon\right)  ^{\gamma}\right)  \right)
\]
is a mild solution of (\ref{F3E2}), (\ref{F3E3}) if it satisfies:%
\begin{equation}
f\left(  t,\epsilon_{1}\right)  =f_{0}\left(  \epsilon_{1}\right)  \Psi\left(
t,\epsilon_{1}\right)  +\frac{8\pi^{2}}{\sqrt{2}}\int_{0}^{t}\frac{\Psi\left(
t,\epsilon_{1}\right)  }{\Psi\left(  s,\epsilon_{1}\right)  }\int_{0}^{\infty
}\int_{0}^{\infty}f_{3}f_{4}\left(  1+f_{1}+f_{2}\right)  Wd\epsilon
_{3}d\epsilon_{4}ds\nonumber
\end{equation}
$a.e.$ $t\in\left[  T_{1},T_{2}\right)  $, where:
\begin{equation}
a\left(  t,\epsilon_{1}\right)  =\frac{8\pi^{2}}{\sqrt{2}}\int_{0}^{\infty
}\int_{0}^{\infty}f_{2}\left(  1+f_{3}+f_{4}\right)  Wd\epsilon_{3}%
d\epsilon_{4}\ \ ,\ \ \Psi\left(  t,\epsilon_{1}\right)  =\exp\left(
-\int_{T_{1}}^{t}a\left(  s,\epsilon_{1}\right)  ds\right)  .\label{F3E6bb}%
\end{equation}

\end{definition}

\begin{remark}
Since $\gamma>3$ there exists a constant $C>0$ such that, for any $t\in\lbrack
T_{1},T_{2})$:
\[
\int_{0}^{\infty}\int_{0}^{\infty}f_{3}f_{4}\left(  1+f_{1}+f_{2}\right)
Wd\epsilon_{3}d\epsilon_{4}\leq C\left\Vert f\right\Vert _{L^{\infty}\left(
\mathbb{R}^{+};\left(  1+\epsilon\right)  ^{\gamma}\right)  }\left(
1+\left\Vert f\right\Vert _{L^{\infty}\left(  \mathbb{R}^{+};\left(
1+\epsilon\right)  ^{\gamma}\right)  }\right)
\]
The term $a\left(  t,\epsilon_{1}\right)  $ is bounded by $C\left(
1+\left\Vert f\right\Vert _{L^{\infty}\left(  \mathbb{R}^{+};\left(
1+\epsilon\right)  ^{\gamma}\right)  }\right)  \int_{0}^{\infty}\int
_{0}^{\infty}f_{2}Wd\epsilon_{3}d\epsilon_{4}$. By the definition of $W$, and
using $\epsilon_{2}$ as one of the integration variables, we estimate
$\int_{0}^{\infty}\int_{0}^{\infty}f_{2}Wd\epsilon_{3}d\epsilon_{4}$ as
\hfill\break$\int_{0}^{\infty}f_{2}\left(  \sqrt{\epsilon_{2}}+\sqrt
{\epsilon_{1}}\right)  Wd\epsilon_{2}\leq C\left\Vert f\right\Vert
_{L^{\infty}\left(  \mathbb{R}^{+};\left(  1+\epsilon\right)  ^{\gamma
}\right)  }.$ Therefore, all the terms in (\ref{F3E6bb}) are well defined\ for
$T_{1}\leq t<T_{2}$ if $f\in L_{loc}^{\infty}\left(  \left[  T_{1}%
,T_{2}\right)  ;L^{\infty}\left(  \mathbb{R}^{+};\left(  1+\epsilon\right)
^{\gamma}\right)  \right)  $.
\end{remark}

\subsection{Local well posedness of mild solutions}.

The following well-posedness result has been proved in \cite{EV1}.

\begin{theorem}
\label{localExistence}Suppose that $f_{0}\in L^{\infty}\left(  \mathbb{R}%
^{+};\left(  1+\epsilon\right)  ^{\gamma}\right)  $ with $\gamma>3.$ There
exists $T>0$ depending only on $\left\Vert f_{0}\left(  \cdot\right)
\right\Vert _{L^{\infty}\left(  \mathbb{R}^{+};\left(  1+\epsilon\right)
^{\gamma}\right)  }$ and a unique mild solution of (\ref{F3E2}), (\ref{F3E3})
in the sense of Definition \ref{mild}, $f\in L_{loc}^{\infty}\left(  \left[
0,T\right)  ;L^{\infty}\left(  \mathbb{R}^{+};\left(  1+\epsilon\right)
^{\gamma}\right)  \right)  $.

The obtained solution $f$ satisfies:%
\begin{equation}
4\pi\sqrt{2}\int_{0}^{\infty}f_{0}\left(  \epsilon\right)  \epsilon
^{w}d\epsilon=4\pi\sqrt{2}\int_{0}^{\infty}f\left(  t,\epsilon\right)
\epsilon^{w}d\epsilon,\ t\in\left(  0,T\right) ,\,w\in\left\{ \frac{1}%
{2},\ \frac{3}{2}\right\} .\label{F3E6c}%
\end{equation}

The function $f$ is in the space $W^{1,\infty}\left(  \left(  0,T\right)
;L^{\infty}\left(  \mathbb{R}^{+}\right)  \right)  $ and it satisfies
(\ref{F3E2}) $a.e.\ \epsilon\in\mathbb{R}^{+}$ for any $t\in\left(  0,T_{\max
}\right)  ,$ with initial datum $f\left(  \cdot,t\right)  =f_{0}\left(
\cdot\right)  .$ Moreover, $f$ can be extended as a mild solution of
(\ref{F3E2}), (\ref{F3E3}) to a maximal time interval $\left(  0,T_{\max
}\right)  $ with $0<T_{\max}\leq\infty.$ If $T_{\max}<\infty$ we have:%
\begin{equation}
\lim\sup_{t\rightarrow T_{\max}^{-}}\left\Vert f\left(  t,\cdot\right)
\right\Vert _{L^{\infty}\left(  \mathbb{R}^{+}\right)  }=\infty\label{F3E6b}%
\end{equation}

\end{theorem}

\subsection{Weak solutions.}

The  Bose-Einstein condensation
phenomena, at the level of the kinetic equation (\ref{F3E2}), (\ref{F3E3}),  
is related in the physical literature
with the onset of a macroscopic fraction of particles at the energy level
$\epsilon=0$. In order to take that phenomena into account it is necessary to consider some notion of weak  solutions for  equation (\ref{F3E2}), (\ref{F3E3}).

The theory of weak solutions of (\ref{F3E2}), (\ref{F3E3}) has been
developed by X. Lu in \cite{Lu1, Lu2, Lu3}. It  allows to deal with
measure-valued solutions and suits very well to the purpose of  considering the finite time formation of Dirac 
mass in the solutions of (\ref{F3E2}), (\ref{F3E3}).

Since we are interested in the condensation
phenomena, it is convenient to use the equation for the mass density $g$,
instead of (\ref{F3E2}), (\ref{F3E3}) for $f$. 
Suppose that $g\left(  t,\epsilon\right)  =4\pi
\sqrt{2\epsilon}f\left(  t,\epsilon\right)  $ (cf. (\ref{F3E3a})). Then,
formally (\ref{F3E2}), (\ref{F3E3}) becomes:
\begin{equation}
\partial_{t}g_{1}=32\pi^{3}\int_{0}^{\infty}\int_{0}^{\infty}q\left(  \frac
{g}{4\pi\sqrt{2\epsilon}}\right)  \Phi d\epsilon_{3}d\epsilon_{4}%
\ \label{F3E4}%
\end{equation}%
\begin{equation}
\Phi=\min\left\{  \sqrt{\epsilon_{1}},\sqrt{\epsilon_{2}},\sqrt{\epsilon_{3}%
},\sqrt{\epsilon_{4}}\right\}  \ \ ,\ \ \epsilon_{2}=\epsilon_{3}+\epsilon
_{4}-\epsilon_{1}\label{F3E5}%
\end{equation}
We will denote as $%
\mathcal{M}_{+}\left( \mathbb{R}^{+};1+\epsilon\right) $ the set of Radon
measures $g$ in $\mathbb{R}^{+}$ satisfying:%
\begin{equation*}
\int\left( 1+\epsilon\right) g\left( \epsilon\right) d\epsilon<\infty 
\end{equation*}
\\
We will use the notation $g\left( \epsilon\right) $ in spite of the fact
that $g$ is a measure. 
\\
 It is natural to say that, for some $T>0$, a weak solution $g$  of
(\ref{F3E4}), (\ref{F3E5}) develops a condensate in finite time $T$,  if the following property holds:
\begin{equation}
\sup_{0<t\leq T}\int_{\left\{  0\right\}  }g\left(  t,\epsilon\right)
d\epsilon>0.\ \label{Z1E7}
\end{equation}
\\
\begin{definition}
\label{weak}We will say that $g\in C\left( \left[ 0,T\right) ;\mathcal{M}%
_{+}\left( \mathbb{R}^{+};\left( 1+\epsilon\right) \right) \right) $ is a
weak solution of (\ref{F3E4}), (\ref{F3E5}) in $(0, T)$, with initial datum $g_{0}\in%
\mathcal{M}_{+}\left( \mathbb{R}^{+};1+\epsilon\right)$,  if, \ for any $%
\varphi\in C_{0}^{2}\left( \left[ 0,T\right) ,\left[ 0,\infty\right) \right)$, the following identity holds:%
\begin{align}
&-\int_{\mathbb{R}^{+}}g_{0}\left( \epsilon\right) \varphi\left(
0,\epsilon\right) d\epsilon  =\int_{0}^{T}\int_{\mathbb{R}%
^{+}}g\partial_{t}\varphi d\epsilon dt+  \notag \\
& \hskip1cm+\frac{1}{2^{\frac{5}{2}}}%
\int_{0}^{T}\int_{\mathbb{R}^{+}}\int_{\mathbb{R}^{+}}\int_{\mathbb{R}^{+}}%
\frac {g_{1}g_{2}g_{3}\Phi}{\sqrt{\epsilon_{1}\epsilon_{2}\epsilon_{3}}}%
Q_{\varphi }d\epsilon_{1}d\epsilon_{2}d\epsilon_{3}dt  +\notag \\
& \hskip2cm+\frac{\pi}{2}\int_{0}^{T}\int_{\mathbb{R}^{+}}\int_{\mathbb{R}%
^{+}}\int_{\mathbb{R}^{+}}\frac{g_{1}g_{2}\Phi}{\sqrt{\epsilon_{1}\epsilon
_{2}}}Q_{\varphi}d\epsilon_{1}d\epsilon_{2}d\epsilon_{3}dt\ \ \ 
\label{Z1E2N}
\end{align}
where $\Phi$ is as in (\ref{F3E5}) and:%
\begin{equation}
Q_{\varphi}=\varphi\left( \epsilon_{3}\right) +\varphi\left( \epsilon
_{1}+\epsilon_{2}-\epsilon_{3}\right) -2\varphi\left( \epsilon_{1}\right) 
\label{Z1E4N}
\end{equation}
\end{definition}
\begin{definition}
\label{weakf}
We say that $f$ is a weak solution of  (\ref{F3E2}), (\ref{F3E3}) in $(0, T)$ with initial datum $f_0$, if  
$g_0=4\pi\sqrt{2\epsilon}f_0\left(\epsilon\right)\in
\mathcal{M}_{+}\left( \mathbb{R}^{+};1+\epsilon\right)$, and  $g=4\pi\sqrt{2\epsilon}f\left( t,\epsilon\right)  
\in C\left( \left[ 0,T\right) ;\mathcal{M}
_{+}\left( \mathbb{R}^{+};\left( 1+\epsilon\right) \right) \right) $ is a
weak solution of (\ref{F3E4}), (\ref{F3E5}) on $(0, T)$, with initial datum $g_{0}$ in the sense of Definition \ref{weak}.
\end{definition}
The following result has been proved proved in \cite{EV1} (Lemma 3.13).
\begin{lemma}
\label{der17Bis}Suppose that $\gamma>3$ and $f\in L_{loc}^{\infty}\left( \left[
0,T\right) ;L^{\infty}\left( \mathbb{R}^{+};\left( 1+\epsilon\right)
^{\gamma}\right) \right) $ is a mild solution of (\ref{F3E2}), (\ref{F3E3})
in the sense of Definition \ref{mild}. Then, $f$
is also a weak solution of  (\ref{F3E2}), (\ref{F3E3}) on $(0, T)$ in the sense of definition \ref{weakf}.
\end{lemma}
It has been proved by X. Lu in Theorem 2 of \cite{Lu1} that for all  $g_{0}\in
\mathcal{M}_{+}\left( \mathbb{R}^{+};1+\epsilon\right)$, there exists a  global weak solution of (\ref{F3E4}), (\ref{F3E5})
 in the sense of the Definition \ref{weak}. This solution satisfies (\ref{Z1E2N}), (\ref{Z1E4N}) for all $\varphi \in C^2_b([0, \infty))$ and with $T=\infty$. Moreover  the two quantities $\int _{ [0, \infty) } g(t, \epsilon )d\epsilon$  and $\int _{ [0, \infty) }\epsilon g(t, \epsilon )d\epsilon$ are constant in time, for all $t\ge 0$.

\section{Statements of the main results.}

\setcounter{equation}{0} \setcounter{theorem}{0}

\subsection{Subcritical and supercritical data. Phase diagram in terms of the
energy and the density of particles.\label{StatSol}}

For isotropic distributions the steady states $F_{BE}$ in (\ref{St1}) have
$p_{0}=0$ and reduce to:%
\begin{align}
& F_{BE}\left(  p\right)  =F_{BE}\left(  p;\alpha,\beta,m_{0}\right)
=m_{0}\delta\left(  p\right)  +\frac{1}{\exp\left(  \beta\left(
\frac{\left\vert p\right\vert ^{2}}{2}+\alpha\right)  \right)  -1}%
\label{St2a}\\
& \text{with }m_{0}\geq0,\ \beta\in( 0,\infty] ,\, 0\leq\alpha<\infty
\ \text{and\ \ }\alpha\cdot m_{0}=0.\label{St2}%
\end{align}

It is customary in the physical literature to denote $\alpha$ as $-\beta\mu,$
where $\mu<0$ is a quantity with units of energy termed as chemical potential.
We will use the notation (\ref{St2}) in order to use nonnegative quantities in
the arguments. If $\alpha=0$ and $m_{0}=0$ the resulting distributions are the
Planck distributions.

The following result shows that at equilibrium there exists a one-to-one
relation between the values of the particle density and the energy and the
values of the chemical potential, temperature and density of particles in the
condensate state.

\begin{proposition}
\label{crit} Given $M>0,\ E>0$ there exists a unique steady state \hfill
\break$F_{BE}\left(  p;\alpha,\beta,m_{0}\right)  $ in the family (\ref{St2})
such that:%
\[
\int_{\mathbb{R}^{3}}F_{BE}\left(  p;\alpha,\beta,m_{0}\right)  d^{3}%
p=M\ \ ,\ \ \int_{\mathbb{R}^{3}}F_{BE}\left(  p;\alpha,\beta,m_{0}\right)
\frac{\left\vert p\right\vert ^{2}}{2}d^{3}p=E\ \
\]

We have $m_{0}=0$ iff
\begin{equation}
M\leq\frac{\zeta\left(  \frac{3}{2}\right)  }{\left(  \zeta\left(  \frac{5}%
{2}\right)  \right)  ^{\frac{3}{5}}}\left(  \frac{4\pi}{3}\right)  ^{\frac
{3}{5}}E^{\frac{3}{5}}\ .\label{Crit}%
\end{equation}
where $\zeta\left(  \cdot\right)  $ is the Riemann function.
\end{proposition}

\begin{proof}
This result is well known and it can be found in classical texts of
Statistical Physics (cf. for instance \cite{Balescu}, \cite{Huang}). We sketch
here the main arguments. Suppose that $m_{0}=0.$ Then, since $\alpha\geq0$ we
have (cf. (\ref{St1})):%
\begin{align}
M  & \leq\int_{\mathbb{R}^{3}}\frac{d^{3}p}{\exp\left(  \frac{\beta\left\vert
p\right\vert ^{2}}{2}\right)  -1}=2\pi\left(  \frac{2}{\beta}\right)
^{\frac{3}{2}}\int_{0}^{\infty}\frac{x^{\frac{1}{2}}dx}{e^{x}-1}%
=\label{E1E1}\\
& =2\pi\left(  \frac{2}{\beta}\right)  ^{\frac{3}{2}}\sum_{n=0}^{\infty}%
\int_{0}^{\infty}x^{\frac{1}{2}}e^{-x}e^{-nx}dx=\left(  \frac{2\pi}{\beta
}\right)  ^{\frac{3}{2}}\zeta\left(  \frac{3}{2}\right)  \ \nonumber
\end{align}
where we have used spherical coordinates, the change of variables
$x=\frac{\beta r^{2}}{2}$ as well as Taylor's series to expand $\left(
1-e^{-x}\right)  ^{-1}.$ A similar computation yields:
\begin{align}
E  & =\frac{1}{2}\int_{\mathbb{R}^{3}}\frac{\left\vert p\right\vert ^{2}%
d^{3}p}{\exp\left(  \frac{\beta\left\vert p\right\vert ^{2}}{2}\right)
-1}=\pi\left(  \frac{2}{\beta}\right)  ^{\frac{5}{2}}\int_{0}^{\infty}%
\frac{x^{\frac{3}{2}}dx}{e^{x}-1}\nonumber\\
& =\pi\left(  \frac{2}{\beta}\right)  ^{\frac{5}{2}}\sum_{n=0}^{\infty}%
\int_{0}^{\infty}x^{\frac{3}{2}}e^{-x}e^{-nx}dx=\frac{3}{4\pi}\left(
\frac{2\pi}{\beta}\right)  ^{\frac{5}{2}}\zeta\left(  \frac{5}{2}\right)
\label{E1E2}%
\end{align}

Then, if $m_{0}=0$ we obtain $M\leq\frac{\zeta\left(  \frac{3}{2}\right)
}{\left(  \zeta\left(  \frac{5}{2}\right)  \right)  ^{\frac{3}{5}}}\left(
\frac{4\pi}{3}\right)  ^{\frac{3}{5}}E^{\frac{3}{5}}.$ On the other hand, if
$m_{0}>0,$ we can argue as in the derivation of (\ref{E1E1}) and using the
fact that $\alpha=0$ we obtain $M>\left(  \frac{2\pi}{\beta}\right)
^{\frac{3}{2}}\zeta\left(  \frac{3}{2}\right)  .$ Combining this formula with
(\ref{E1E1}) it then follows that $M>\frac{\zeta\left(  \frac{3}{2}\right)
}{\left(  \zeta\left(  \frac{5}{2}\right)  \right)  ^{\frac{3}{5}}}\left(
\frac{4\pi}{3}\right)  ^{\frac{3}{5}}E^{\frac{3}{5}},$ whence the Proposition follows.
\end{proof}

\begin{definition}
We will say that a couple of values $\left(  M,E\right) $, such that $M>0$ and
$E>0$, is in the subcritical region if (\ref{Crit}) holds. On the contrary, if
$M>\frac{\zeta\left(  \frac{3}{2}\right)  }{\left(  \zeta\left(  \frac{5}%
{2}\right)  \right)  ^{\frac{3}{5}}}\left(  \frac{4\pi}{3}\right)  ^{\frac
{3}{5}}E^{\frac{3}{5}}$ we will say that $\left(  M,E\right)  $ is in the
supercritical region.
\end{definition}

\subsection{Blow-up Theorem for supercritical data.}

The main result of this paper about the finite time blow up of solutions is the following.

\begin{theorem}
\label{Main}Suppose that $f_{0}\in L^{\infty}\left(  \mathbb{R}^{+};\left(
1+\epsilon\right)  ^{\gamma}\right)  $ with $\gamma>3.$ Let us denote as
$M,\ E$ the numbers:%
\begin{equation}
4\pi\int_{0}^{\infty}f_{0}\left(  \epsilon\right)  \sqrt{2\epsilon}%
d\epsilon=M\ \ ,\ \ 4\pi\int_{0}^{\infty}f_{0}\left(  \epsilon\right)
\sqrt{2\epsilon^{3}}d\epsilon=E\label{G1E1}%
\end{equation}

Let us denote as $f\in L_{loc}^{\infty}\left(  \left[  0,T_{\max}\right)
;L^{\infty}\left(  \mathbb{R}^{+};\left(  1+\epsilon\right)  ^{\gamma}\right)
\right)  $ the mild solution of (\ref{F3E2}), (\ref{F3E3}) in Theorem
\ref{localExistence} where $T_{\max}$ is the maximal existence time. Suppose
that:%
\begin{equation}
M>\frac{\zeta\left(  \frac{3}{2}\right)  }{\left(  \zeta\left(  \frac{5}%
{2}\right)  \right)  ^{\frac{3}{5}}}\left(  \frac{4\pi}{3}\right)  ^{\frac
{3}{5}}E^{\frac{3}{5}}\ \label{G1E2}%
\end{equation}
Then:%
\begin{equation}
T_{\max}<\infty\label{G1E3}%
\end{equation}

\end{theorem}

\begin{remark}
Notice that $M,$ $E$ are the particle density and energy density associated to
the distribution $f_{0}.$ The numerical factors are due to the change of
variables to spherical coordinates and to the fact that we use $\epsilon$
instead of $p$ as an integration variable.
\end{remark}

Notice that Theorems \ref{localExistence} and Theorem \ref{Main} imply that
for supercritical values of $\left(  M,E\right)  ,$ $f$ blows-up in finite time:

\begin{corollary}
Suppose that the conditions of Theorem \ref{Main} hold. Then:%
\[
\lim\sup_{t\rightarrow T_{\max}^{-}}\left\Vert f\left(  t,\cdot\right)
\right\Vert _{L^{\infty}\left(  \mathbb{R}^{+}\right)  }=\infty
\ \ ,\ \ T_{\max}<\infty
\]

\end{corollary}

\subsection{Condensation Theorem for supercritical data.}

The main result of this paper about the condensation  of solutions in finite time  is the following.

\begin{theorem}
\label{Cond1}Let  $g_{0}\in\mathcal{M}_{+}\left(  \mathbb{R}^{+};\left(  1+\epsilon\right)\right)  $ satisfying
\begin{align}
4\pi\sqrt{2}\int_{\mathbb{R}^{+}}g_{0}\left(  \epsilon\right)  d\epsilon &
=M\ ,\ \ 4\pi\sqrt{2}\int_{\mathbb{R}^{+}}g_{0}\left(  \epsilon\right)
\epsilon d\epsilon=E.\ \label{Z1E5}
\end{align}
Suppose that $M$ and $E$ satisfy condition (\ref{G1E2}). Then, there exists $T_0=T_0(M, E)>0$ such that every weak solution $g$  of  (\ref{F3E4}), (\ref{F3E5})  in $(0, \infty)$ with initial data $g_0$  satisfies  (\ref{Z1E7}) for any $T\ge T_0$.
\end{theorem}

\section{Summary of some results in \cite{EV1}.}

\setcounter{equation}{0} \setcounter{theorem}{0}

We now recall some results which has been obtained in \cite{EV1} which will be
used in the Proof of the results of this paper.

\subsection{Blow-up condition.}

One of the main ingredients in the Proof of Theorem \ref{Main} is the
following  condition for blow-up of mild solutions which has been proved in
\cite{EV1}.

\begin{theorem}
\label{main} There exist $\
\theta_{\ast}>0$ such that,  for all  $M>0,\ E>0,$ $\nu>0,\ \gamma>3$,  there exists $ 
\rho_0=\rho_0\left( M,E,\nu\right) >0$, $ K^{\ast}=K^{\ast}\left( M,E,\nu\right)
>0,\ T_{0}=T_{0}\left( M,E\right) $ satisfying the following property.  For any $f_{0}\in
L^{\infty}\left( \mathbb{R}^{+};\left( 1+\epsilon\right) ^{\gamma}\right) $
such that$\ $%
\begin{align}
4\pi\sqrt{2}\int_{\mathbb{R}^{+}}f_{0}\left( \epsilon\right) \sqrt{\epsilon }%
d\epsilon & =M\ ,\ \ 4\pi\sqrt{2}\int_{\mathbb{R}^{+}}f_{0}\left(
\epsilon\right) \sqrt{\epsilon^{3}}d\epsilon=E,  \label{C1} 
\end{align}
and
 \begin{equation}
 \label{Condicion}
 \sup _{ 0\le \rho \le \rho _0}\left[
 \min\left\{
\inf _{ 0\le R\le \rho  }\frac {1} {\nu R^{3/2}}\int _0^R f_0(\epsilon)\sqrt \epsilon d\epsilon,  \frac {1} {K^*\rho ^{\theta_*}}\int _0^\rho f_0(\epsilon)\sqrt \epsilon d\epsilon
 \right\}
 \right]\ge 1,
 \end{equation}
 there exists a unique $f\in L_{loc}^{\infty}\left( \left[
0,T_{\max}\right) ;L^{\infty}\left( \mathbb{R}^{+};\left( 1+\epsilon \right)
^{\gamma}\right) \right)
$ mild solution  of   (\ref{F3E2}), (\ref{F3E3}) in the sense of Definition \ref{mild}, 
with initial data $f_0$,  defined for a maximal existence time $%
T_{\max}<T_{0}$ and that satisfies:
\begin{equation*}
\limsup_{t\rightarrow T_{\max}^{-}}\left\Vert f\left( \cdot,t\right)
\right\Vert _{L^{\infty}\left( \mathbb{R}^{+}\right) }=\infty. 
\end{equation*}
\end{theorem}
The above Theorem  means that initial data  $f_{0}\in
L^{\infty}\left( \mathbb{R}^{+};\left( 1+\epsilon\right) ^{\gamma}\right) $, with a sufficiently large density around $\epsilon =0$, blows up in finite time.
More precisely, the condition (\ref{Condicion}) means that  there exists $\rho \in(0, \rho _0)$ satisfying:
\begin{align}
\int_{0}^{R}f_{0}\left( \epsilon\right) \sqrt{\epsilon}d\epsilon & \geq\nu
R^{\frac{3}{2}}\ \ \text{for }0<R\leq\rho\ \ ,\ \ \int_{0}^{\rho}f_{0}\left(
\epsilon\right) \sqrt{\epsilon}d\epsilon\geq K^{\ast}\left( \rho\right)
^{\theta_{\ast}},   \label{C2}
\end{align} 
The second condition in (\ref{C2}) holds if the distribution $f_0$ has a mass sufficiently large in a ball with radius $\rho$ for some $\rho$ sufficiently small. The first condition is satisfied if $f_0(\epsilon)\ge 3\nu/2$ for all $\epsilon $ sufficiently small. Since $\theta_*$ might  be small, the first condition in (\ref{C2}) does not implies the second.
\\
Our results  do not provide an explicit functional relation for the functions $\rho _0(M, E, \nu)$,  $K^*(M, E, \nu)$ and $T_0(M, E)$ in terms of their arguments. Therefore, for a given initial data $f_0 \in
L^{\infty}\left( \mathbb{R}^{+};\left( 1+\epsilon\right) ^{\gamma}\right) $ it is not easy to check if condition (\ref{C2}) is satisfied. However, it has been verified in \cite{EV1} that the class of functions $f_0 \in
L^{\infty}\left( \mathbb{R}^{+};\left( 1+\epsilon\right) ^{\gamma}\right) $ satisfying such condition is not empty.

\subsection{Condensation condition.}
The following result about finite
time condensation  for weak solutions of (\ref{F3E2}), (\ref{F3E3})  has been proved  in \cite{EV1}.
\begin{theorem}
\label{Cond1} There exist $\
\theta_{\ast}>0$ with the following property. For all  $M>0,\ E>0,$ $\nu>0$,  there exists $ 
\rho_0=\rho_0\left( M,E,\nu\right) >0$, $ K^{\ast}=K^{\ast}\left( M,E,\nu\right)
>0,\ T_{0}=T_{0}\left( M,E\right) $ such that for any  weak solution of (\ref{F3E4}%
), (\ref{F3E5}) on $(0, T_0)$ in the sense of Definition \ref{weak} with $g_{0}\in\mathcal{%
M}_{+}\left( \mathbb{R}^{+};\left( 1+\epsilon\right) \right) $ satisfying$\ $%
\begin{eqnarray}
&&4\pi\sqrt{2}\int_{\mathbb{R}^{+}}g_{0}\left( \epsilon\right) d\epsilon  =M\
,\ \ 4\pi\sqrt{2}\int_{\mathbb{R}^{+}}g_{0}\left( \epsilon\right) \epsilon
d\epsilon=E\   \label{Z1E5N} \\
&& \sup _{ 0\le \rho \le \rho _0}\left[
 \min\left\{
\inf _{ 0\le R\le \rho  }\frac {1} {\nu R^{3/2}}\int _0^R g_0(\epsilon) d\epsilon,  \frac {1} {K^*\rho ^{\theta_*}}\int _0^\rho g_0(\epsilon) d\epsilon
 \right\}
 \right]\ge 1,  \label{Z1E6}
\end{eqnarray}
we have:
\begin{equation}
\sup_{0<t\leq T_{0}}\int_{\left\{ 0\right\} }g\left( t,\epsilon\right)
d\epsilon>0.\   \label{Z1E7N}
\end{equation}
\end{theorem}

\subsection{Transfer of mass formula.}

We will use the following result,  proved  in \cite{EV1} (Proposition 4.1), to estimate the rate of transfer of particles
between different regions.

We recall that the density of particles for unit of energy is given by (cf.
\cite{EV1}):
\begin{equation}
g\left(  t,\epsilon\right)  =4\pi\sqrt{2\epsilon}f\left(  t,\epsilon\right)
\label{F3E3a}%
\end{equation}

\begin{proposition}
\label{der17}Suppose that $g$ is a weak solution of (\ref{F3E4}), (\ref{F3E5}) in $(0, T)$, with initial datum $g_{0}\in%
\mathcal{M}_{+}\left( \mathbb{R}^{+};1+\epsilon\right)$ in the sense of Definition \ref{weak}. Suppose that $\varphi\in C_{0}^{1}\left(  \left[
0,T\right)  \times\left[  0,\infty\right)  \right)$.  Then, the function
$\psi_{\varphi}\left(  t\right)  =\int_{\mathbb{R}^{+}}g\left(  t,\epsilon
\right)  \varphi\left(  t,\epsilon\right)  d\epsilon$ is Lipschitz continuous
in $t\in\left[  0,T\right]  $ and the following identity holds:
\begin{align}
\partial_{t}\left(  \int_{\mathbb{R}^{+}}\!\!\!\!g\varphi d\epsilon\right)
&  =\int_{\mathbb{R}^{+}}\!\!\!\!g\partial_{t}\varphi d\epsilon+\frac
{1}{2^{\frac{5}{2}}}\int_{\mathbb{R}^{+}}\int_{\mathbb{R}^{+}}\int
_{\mathbb{R}^{+}}\!\frac{g_{1}g_{2}g_{3}\Phi}{\sqrt{\epsilon_{1}\epsilon
_{2}\epsilon_{3}}}\mathcal{G}_{\varphi}\left(  \mathcal{\epsilon}%
_{1},\mathcal{\epsilon}_{2},\mathcal{\epsilon}_{3}\right)  d\epsilon
_{1}d\epsilon_{2}d\epsilon_{3}+\nonumber\\
&  +\frac{\pi}{2}\int_{\mathbb{R}^{+}}\int_{\mathbb{R}^{+}}\int_{\mathbb{R}%
^{+}}\frac{g_{1}g_{2}\Phi}{\sqrt{\epsilon_{1}\epsilon_{2}}}\mathcal{H}%
_{\varphi}d\epsilon_{1}d\epsilon_{2}d\epsilon_{3}\ \ ,\ \ a.e.\ t\in\left[
0,T\right]  \label{F5E1a}%
\end{align}
with:%
\begin{equation}
\Phi=\min\left\{  \sqrt{\epsilon_{1}},\sqrt{\epsilon_{2}},\sqrt{\epsilon_{3}%
},\sqrt{\left(  \epsilon_{1}+\epsilon_{2}-\epsilon_{3}\right)  _{+}}\right\}
\ \label{F5E1b}%
\end{equation}%
\begin{equation}
\mathcal{H}_{\varphi}=\varphi\left(  \epsilon_{3}\right)  +\varphi\left(
\epsilon_{1}+\epsilon_{2}-\epsilon_{3}\right)  -\varphi\left(  \epsilon
_{1}\right)  -\varphi\left(  \epsilon_{2}\right)  \label{F5E1c}%
\end{equation}%
\[
\mathcal{G}_{\varphi}\left(  \mathcal{\epsilon}_{1},\mathcal{\epsilon}%
_{2},\mathcal{\epsilon}_{3}\right)  =\frac{1}{6}\sum_{\sigma\in\mathcal{S}%
^{3}}H_{\varphi}\left(  \mathcal{\epsilon}_{\sigma(1)},\mathcal{\epsilon
}_{\sigma(2)},\mathcal{\epsilon}_{\sigma(3)}\right)  \Phi\left(
\mathcal{\epsilon}_{\sigma(1)},\mathcal{\epsilon}_{\sigma(2)}%
;\mathcal{\epsilon}_{\sigma(3)}\right)
\]
where $\mathcal{S}^{3}$ is the group of permutations of $\left\{
1,2,3\right\}  .$ Moreover, for any convex function $\varphi$ we have:%
\begin{align}
& \mathcal{G}_{\varphi}\left(  \mathcal{\epsilon}_{1},\mathcal{\epsilon}%
_{2},\mathcal{\epsilon}_{3}\right)  \geq0\ \ ,\ \ \ \left(  \mathcal{\epsilon
}_{1},\mathcal{\epsilon}_{2},\mathcal{\epsilon}_{3}\right)  \in\mathbb{R}%
_{+}^{3}\label{F5E0}\\
& \mathcal{G}_{\varphi}\left(  \mathcal{\epsilon}_{1},\mathcal{\epsilon}%
_{2},\mathcal{\epsilon}_{3}\right)  =\mathcal{G}_{\varphi}\left(
\mathcal{\epsilon}_{\sigma(1)},\mathcal{\epsilon}_{\sigma(2)}%
,\mathcal{\epsilon}_{\sigma(3)}\right)  \ \ \ \text{for any\ }\sigma
\in\mathcal{S}^{3}\label{F5E0bis}%
\end{align}
\end{proposition}

\begin{proof}
It is just a consequence of Lemma 3.13 and Theorem 4.2 in \cite{EV1}.
\end{proof}

\section{A lower estimate for the mass in $\left[  0,R\right] $ with $R$
small.}

\setcounter{equation}{0} \setcounter{theorem}{0}

\begin{proposition}
\label{lowMass} 
Suppose that $g$ is a weak solution of (\ref{F3E4}), (\ref{F3E5}) in $(0, T)$, with initial datum $g_{0}\in%
\mathcal{M}_{+}\left( \mathbb{R}^{+};1+\epsilon\right)$ in the sense of Definition \ref{weak}.
Then, there exist $T_{1}%
=T_{1}\left(  E,M\right)  ,\ \rho_{1}=\rho_{1}\left(  E,M\right)  $ and $K=K\left(
E,M\right)  $ such that if $T_{\max}\geq T_{1}\left(  E,M\right)  $, we have:%
\begin{equation}
\label{S5EReg}
\int_{0}^{R}g\left(  \epsilon,t\right)  d\epsilon\geq KR^{\frac{3}{2}%
}\ \ \text{for any }0<R\leq\rho_{1}(E, M)
\end{equation}
for any $t\geq T_{1}\left(  E,M\right)  .$
\end{proposition}

\begin{proof}
We will assume that $\rho_{\textcolor{red}{1}}\leq1.$ Given $0<R\leq1,$ we define the following
family of test functions:%
\[
\varphi_{R}\left(  \epsilon\right)  =\left(  1-\frac{\epsilon}{R}\right)  _{+}%
\]
where $\left(  s\right)  _{+}=\max\left\{  s,0\right\}  .$ Using Proposition
\ref{der17} we obtain:%
\[
\partial_{t}\left(  \int g\varphi_{R}d\epsilon\right)  \geq\frac{\pi}{2}%
\int_{\mathbb{R}^{+}}\int_{\mathbb{R}^{+}}\int_{\mathbb{R}^{+}}\frac
{g_{1}g_{2}\Phi}{\sqrt{\epsilon_{1}\epsilon_{2}}}H_{\varphi_{R}}d\epsilon
_{1}d\epsilon_{2}d\epsilon_{3}\ \ ,\ \ a.e.\ t\in\left[  0,T\right]
\]
where we use the fact that $\varphi_{R}$ is convex, and then (\ref{F5E0})
holds. Using (\ref{F5E1c}) as well as the fact that $\varphi\left(
\epsilon_{1}+\epsilon_{2}-\epsilon_{3}\right)  \geq0$ we obtain for
$a.e.\ t\in\left[  0,T\right]  :$%
\begin{equation}
\partial_{t}\left(  \int g\varphi_{R}d\epsilon\right)  \geq I_{1}%
-I_{2}\ \label{G1E6}%
\end{equation}
where%
\begin{align}
I_{1}  & =\frac{\pi}{2}\int_{\mathbb{R}^{+}}\int_{\mathbb{R}^{+}}%
\int_{\mathbb{R}^{+}}\frac{g_{1}g_{2}\Phi}{\sqrt{\epsilon_{1}\epsilon_{2}}%
}\varphi_{R}\left(  \epsilon_{3}\right)  d\epsilon_{1}d\epsilon_{2}%
d\epsilon_{3}\label{G1E7}\\
I_{2}  & =\frac{\pi}{2}\int_{\mathbb{R}^{+}}\int_{\mathbb{R}^{+}}%
\int_{\mathbb{R}^{+}}\frac{g_{1}g_{2}\Phi}{\sqrt{\epsilon_{1}\epsilon_{2}}%
}\left[  \varphi_{R}\left(  \epsilon_{1}\right)  +\varphi_{R}\left(
\epsilon_{2}\right)  \right]  d\epsilon_{1}d\epsilon_{2}d\epsilon
_{3}\label{G1E8}%
\end{align}

We estimate $I_{2}$, using the definition of $\Phi,$ as well as the symmetry
with respect to the permutations of $\epsilon_{1},\epsilon_{2},$ as follows:%
\begin{align*}
I_{2}  & =\pi\int_{\mathbb{R}^{+}}\int_{\mathbb{R}^{+}}\int_{\mathbb{R}^{+}%
}\frac{g_{1}g_{2}\Phi}{\sqrt{\epsilon_{1}\epsilon_{2}}}\varphi_{R}\left(
\epsilon_{1}\right)  d\epsilon_{1}d\epsilon_{2}d\epsilon_{3}\\
& \leq\pi\int_{0}^{R}g_{1}\varphi_{R}\left(  \epsilon_{1}\right)
d\epsilon_{1}\int_{\mathbb{R}^{+}}g_{2}d\epsilon_{2}\left[  \frac{\min\left\{
\sqrt{\epsilon_{1}},\sqrt{\epsilon_{2}}\right\}  \left(  \epsilon_{1}%
+\epsilon_{2}\right)  }{\sqrt{\epsilon_{1}\epsilon_{2}}}\right]
\end{align*}
where we use that $\Phi\leq\min\left\{  \sqrt{\epsilon_{1}},\sqrt{\epsilon
_{2}}\right\}  $ and we compute the integral in $\epsilon_{3}$. Notice that
\hfill\break$\frac{\min\left\{  \sqrt{\epsilon_{1}},\sqrt{\epsilon_{2}%
}\right\}  \left(  \epsilon_{1}+\epsilon_{2}\right)  }{\sqrt{\epsilon
_{1}\epsilon_{2}}}\leq2\max\left\{  \sqrt{\epsilon_{1}},\sqrt{\epsilon_{2}%
}\right\}  .$ Then, since $\epsilon_{1}\leq R\leq1:$%
\begin{align*}
\int_{\mathbb{R}^{+}}g_{2}\left[  \frac{\min\left\{  \sqrt{\epsilon_{1}}%
,\sqrt{\epsilon_{2}}\right\}  \left(  \epsilon_{1}+\epsilon_{2}\right)
}{\sqrt{\epsilon_{1}\epsilon_{2}}}\right]  d\epsilon_{2}  & \leq
2\int_{\mathbb{R}^{+}}g_{2}\max\left\{  1,\sqrt{\epsilon_{2}}\right\}
d\epsilon_{2}\\
& \leq2\int_{\mathbb{R}^{+}}g_{2}\left(  1+\epsilon_{2}\right)  d\epsilon
_{2}\leq2\left(  M+E\right)
\end{align*}%
\begin{equation}
I_{2}\leq2\pi\left(  M+E\right)  \int_{0}^{R}g\varphi_{R}d\epsilon\label{G2E3}%
\end{equation}

We now estimate $I_{1}$ assuming that $\int_{0}^{R}g\varphi_{R}d\epsilon
\leq\frac{M}{4}.$ Using the fact that $\int gd\epsilon=M$ we have:
\begin{equation}
\int_{\frac{R}{2}}^{\infty}gd\epsilon\geq\frac{M}{8}\label{G1E9}%
\end{equation}

Indeed, suppose that $\int_{\frac{R}{2}}^{\infty}gd\epsilon<\frac{M}{8},$ then
$\int_{0}^{\frac{R}{2}}gd\epsilon\geq\frac{7M}{8},$ whence, since $\varphi
_{R}\geq\frac{1}{2}$ for $0\leq\epsilon\leq\frac{R}{2},$ we have $\int_{0}%
^{R}g\varphi_{R}d\epsilon\geq\frac{7M}{16}>\frac{M}{4}$ against the
assumption, whence (\ref{G1E9}) follows.

Using the definition of $\Phi$ (cf. (\ref{F5E1b})) we obtain:%
\begin{equation}
I_{1}=\frac{\pi R^{\frac{3}{2}}}{2}\int_{\mathbb{R}^{+}}\int_{\mathbb{R}^{+}%
}\frac{J\left(  \frac{\epsilon_{1}}{R},\frac{\epsilon_{2}}{R}\right)  }%
{\sqrt{\epsilon_{1}\epsilon_{2}}}g_{1}g_{2}d\epsilon_{1}d\epsilon
_{2}\ \label{G1E4}%
\end{equation}
where:%
\begin{equation}
J\left(  \eta_{1},\eta_{2}\right)  =\int_{0}^{\eta_{1}+\eta_{2}}\min\left\{
\sqrt{\eta_{1}},\sqrt{\eta_{2}},\sqrt{\eta_{3}}, \sqrt{\left(  \eta_{1}%
+\eta_{2}-\eta_{3}\right)  _{+}}\right\}  \left(  1-\eta_{3}\right)  _{+}%
d\eta_{3}\label{G1E5}%
\end{equation}

Notice that, if $\eta_{1}\geq\frac{1}{2},\ \eta_{2}\geq\frac{1}{2}$ we have
$J\left(  \eta_{1},\eta_{2}\right)  \geq c_{0}>0$ where $c_{0}$ is a number
independent on $R.$ Then:%
\begin{align}
I_{1} & \geq \frac{\pi R^{\frac{3}{2}}}{2}\int_{\frac{R}{2}}^{\infty}%
\int_{\frac{R}{2}}^{\infty}\frac{J\left(  \frac{\epsilon_{1}}{R}%
,\frac{\epsilon_{2}}{R}\right)  }{\sqrt{\epsilon_{1}\epsilon_{2}}}g_{1}%
g_{2}d\epsilon_{1}d\epsilon_{2}\nonumber\\
& \geq \frac{\pi c_{0}R^{\frac{3}{2}}}{2}\int_{\frac{R}{2}}^{\infty}%
\int_{\frac{R}{2}}^{\infty}\frac{g_{1}g_{2}}{\sqrt{\epsilon_{1}\epsilon_{2}}%
}d\epsilon_{1}d\epsilon_{2}=\frac{\pi c_{0}R^{\frac{3}{2}}}{2}\left(
\int_{\frac{R}{2}}^{\infty}\frac{g}{\sqrt{\epsilon}}d\epsilon\right)
^{2}\label{G2E1}%
\end{align}

Due to the boundedness of the energy we have:%
\[
\int_{L}^{\infty}gd\epsilon\leq\frac{1}{L}\int_{L}^{\infty}g\epsilon
d\epsilon\leq\frac{E}{L}%
\]

Then, if $L\geq\frac{16E}{M},$ $\int_{L}^{\infty}gd\epsilon\leq\frac{M}{16},$
whence, using (\ref{G1E9}):%
\[
\int_{\frac{R}{2}}^{L}gd\epsilon\geq\frac{M}{16}%
\]

Therefore (\ref{G2E1}) yields:%
\[
I_{1}\geq\frac{\pi c_{0}R^{\frac{3}{2}}}{2}\left(  \int_{\frac{R}{2}}^{L}%
\frac{g}{\sqrt{\epsilon}}d\epsilon\right)  ^{2}\geq\frac{\pi c_{0}R^{\frac
{3}{2}}}{2L}\left(  \int_{\frac{R}{2}}^{L}gd\epsilon\right)  ^{2}\geq\frac{\pi
c_{0}M^{2}R^{\frac{3}{2}}}{512L}%
\]

We have then obtained that there exists $\alpha=\alpha\left(  E,M\right)  $
such that, \hfill\break if $\int_{0}^{R}g\varphi_{R}d\epsilon\leq\frac{M}{4}:$%
\begin{equation}
I_{1}\geq\alpha R^{\frac{3}{2}}.\label{G2E2}%
\end{equation}

Combining (\ref{G1E6}), (\ref{G2E3}) and (\ref{G2E2}) we obtain:%
\begin{equation}
\partial_{t}\left(  \int g\varphi_{R}d\epsilon\right)  \geq\alpha R^{\frac
{3}{2}}-2\pi\left(  M+E\right)  \int_{0}^{R}g\varphi_{R}d\epsilon
\ \ \ \text{if }\int_{0}^{R}g\varphi_{R}d\epsilon\leq\frac{M}{4}\label{G2E4}%
\end{equation}

Choosing $R\leq\rho_{1}<\left(  \frac{\pi\left(  M+E\right)  M}{\alpha}\right)
^{\frac{2}{3}}$ it follows that, for $\int_{0}^{R}g\varphi_{R}d\epsilon
\leq\frac{\alpha R^{\frac{3}{2}}}{4\pi\left(  M+E\right)  }$ we would have
also $\int_{0}^{R}g\varphi_{R}d\epsilon<\frac{M}{4}$ and the inequality
(\ref{G2E4}) holds. This implies that $\int g\varphi_{R}d\epsilon$ increases
exponentially until reaching the value $\frac{\alpha R^{\frac{3}{2}}}%
{4\pi\left(  M+E\right)  }$ or larger in a time which depends only on $M$ and
$E$. If $\int_{0}^{R}g\varphi_{R}d\epsilon>\frac{M}{4}$ the inequality $\int
g\varphi_{R}d\epsilon\geq\frac{\alpha R^{\frac{3}{2}}}{4\pi\left(  M+E\right)
}$ is automatically satisfied, due to our choice of $\rho_{1}$. On the other hand,
if $\frac{\alpha R^{\frac{3}{2}}}{4\pi\left(  M+E\right)  }\leq\int_{0}%
^{R}g\varphi_{R}d\epsilon\leq\frac{M}{4}$ the value of $\int_{0}^{R}%
g\varphi_{R}d\epsilon$ remains in that region due to (\ref{G2E4}). Therefore,
there exists $T_{1}=T_{1}\left(  M,E\right)  $ such that if $T_{1}\left(
M,E\right)  \leq T_{\max}$:%
\[
\frac{\alpha R^{\frac{3}{2}}}{4\pi\left(  M+E\right)  }\leq\int g\varphi
_{R}d\epsilon\leq\int_{0}^{R}gd\epsilon\ ,\ \text{for\ }t\geq T_{1}\left(
M,E\right)
\]
whence the result follows.
\end{proof}

\section{Dissipation of entropy formula.}

\setcounter{equation}{0} \setcounter{theorem}{0} We now formulate the formula
for the dissipation of the entropy. The entropy associated to the bosonic,
isotropic quantum Boltzmann equation (\ref{F3E2}), (\ref{F3E3}) is given by:%
\begin{equation}
S\left[  f\right]  =\int_{\mathbb{R}^{+}}\left[  \left(  1+f\right)
\log\left(  1+f\right)  -f\log\left(  f\right)  \right]  \sqrt{\epsilon
}d\epsilon\label{G2E5}%
\end{equation}

We also define, for any given $f$ the dissipation of entropy by means of:%
\begin{align}
D\left[  f\right]   & =\int_{\mathbb{R}^{+}}\int_{\mathbb{R}^{+}}%
\int_{\mathbb{R}^{+}}\left(  1+f_{1}\right)  \left(  1+f_{2}\right)  \left(
1+f_{3}\right)  \left(  1+f_{4}\right)  \left[  Q_{1,2}-Q_{3,4}\right]
\times\nonumber\\
& \times\left[  \log\left(  Q_{1,2}\right)  -\log\left(  Q_{3,4}\right)
\right]  \Phi d\epsilon_{1}d\epsilon_{2}d\epsilon_{3}\ \label{G2E6}%
\end{align}
where:%
\begin{align*}
Q_{j,k}  & =\frac{f_{j}}{\left(  1+f_{j}\right)  }\frac{f_{k}}{\left(
1+f_{k}\right)  }\ \ ,\ \ j,k\in\left\{  1,2,3,4\right\}  \ \ \ \\
\epsilon_{4}  & =\epsilon_{1}+\epsilon_{2}-\epsilon_{3}\ \ ,\ \ \Phi
=\min\left\{  \sqrt{\left(  \epsilon_{1}\right)  _{+}},\sqrt{\left(
\epsilon_{2}\right)  _{+}},\sqrt{\left(  \epsilon_{3}\right)  _{+}}%
,\sqrt{\left(  \epsilon_{4}\right)  _{+}}\right\}
\end{align*}

Formulas for the dissipation of the entropy for suitably truncated kernels $W$ and solutions in $L^1$ have been proved in \cite{Lu2}). Since we are working here with mild solutions we have decided to write a new proof of the dissipation of energy formula.
\begin{proposition}
\label{Entropy} Suppose that $f_{0}\in L^{\infty}\left(  \mathbb{R}%
^{+};\left(  1+\epsilon\right)  ^{\gamma}\right)  $ with $\gamma>3,\ $ and
$f_{0}\not \equiv 0.$ Suppose that $f\in L_{loc}^{\infty}\left(  \left[
0,T_{\max}\right)  ;L^{\infty}\left(  \mathbb{R}^{+};\left(  1+\epsilon
\right)  ^{\gamma}\right)  \right)  $ is the mild solution of (\ref{F3E2}),
(\ref{F3E3}) with initial datum $f_{0}$ as in Theorem \ref{main}.
Then, there exists $C=C\left(  E,M\right)  $ such that:%
\begin{equation}
\left\vert S\left[  f\left(  \cdot,t\right)  \right]  \right\vert \leq
C\left(  E,M\right)  \ \ ,\ 0\leq t<T_{\max}\label{G4E7a}%
\end{equation}
For any $T_{1}$ and $T_{2}$ such that $2 T_{0}<T_{1}<T_{2}\leq T_{\max}$,
where $T_{0}$ as in proposition (\ref{lowMass}), we have:%
\[
0\leq D\left[  f\left(  \cdot,t\right)  \right]  \leq C\left(  T_{1}%
,T_{2}\right)  <\infty\ \ \ \ \ \ ,\ \ \ \ T_{1}\leq t\leq T_{2}%
\]
for some constant $C\left(  T_{1},T_{2}\right)  .$ Moreover, the following
identity holds:%
\begin{equation}
S\left[  f\right]  \left(  T_{2}\right)  -S\left[  f\right]  \left(
T_{1}\right)  =\int_{T_{1}}^{T_{2}}D\left[  f\left(  \cdot,t\right)  \right]
dt\label{G4E7}%
\end{equation}

\end{proposition}

\begin{proof}
As a first step we prove (\ref{G4E7a}). To this end we use the following
inequalities:
\begin{align*}
\left[  \left(  1+f\right)  \log\left(  1+f\right)  -f\log\left(  f\right)
\right]   & \leq Cf\ \ \text{if\ \ }f\geq1\\
\left[  \left(  1+f\right)  \log\left(  1+f\right)  -f\log\left(  f\right)
\right]   & \leq Cf\left\vert \log\left(  f\right)  \right\vert
\ \ \text{if\ \ }0\leq f\leq1
\end{align*}

We then can estimate $S\left[  f\right]  $ for any $f\in L^{\infty}\left(
\mathbb{R}^{+};\left(  1+\epsilon\right)  ^{\gamma}\right)  ,$ $\gamma>3$ as
follows. We have the inequalities:%
\begin{align*}
S\left[  f\right]   & =\int_{\left\{  f\geq1\right\}  }\left[  \cdot\cdot
\cdot\right]  \sqrt{\epsilon}d\epsilon+\int_{\left\{  f<1,\ \left\vert
\log\left(  f\right)  \right\vert <\sqrt{\epsilon}\right\}  }\left[
\cdot\cdot\cdot\right]  \sqrt{\epsilon}d\epsilon+\int_{\left\{
f<1,\ \left\vert \log\left(  f\right)  \right\vert \geq\sqrt{\epsilon
}\right\}  }\left[  \cdot\cdot\cdot\right]  \sqrt{\epsilon}d\epsilon\\
& \leq C\int f\sqrt{\epsilon}d\epsilon+C\int f\epsilon^{\frac{3}{2}}%
d\epsilon+C\int_{\left\{  f<\exp\left(  -\epsilon\right)  \right\}  }\sqrt
{f}\left(  \sqrt{f}\left\vert \log\left(  f\right)  \right\vert \right)
\sqrt{\epsilon}d\epsilon\\
& \leq C\left(  M+E\right)  +C\int\exp\left(  -\frac{\epsilon}{2}\right)
d\epsilon=C\left(  M,E\right)
\end{align*}
whence (\ref{G4E7a}) follows.

We now need a lower estimate for $f\left(  \epsilon,t\right)  ,\ t>0$ in order
to ensure the convergence of the integrals in (\ref{G2E6}). To this end we use
methods inspired in the ideas used by Carleman to prove the $H-$Theorem for
the classical Boltzmann equation (cf. \cite{Ca1}, Section 5). We claim that
for any $\omega>0$ and $0<T_{1}\leq T_{2},$ there exists $C_{T_{1}%
,T_{2},\omega}$ such that:%
\begin{equation}
f\left(  \epsilon,t\right)  \geq C_{T_{1},T_{2},\omega}\exp\left(
-\epsilon^{1+\omega}\right)  \ \ ,\ T_{1}\leq t\leq T_{2}\label{E1E3}%
\end{equation}
where the constant $C_{T_{1},T_{2},\omega}$ depends also on $f_{0}.$

In order to obtain (\ref{E1E3}), we first remark that the mild solutions of
(\ref{F3E2}), (\ref{F3E3}) satisfy (cf. \cite{EV1}, Lemma 3.7):%

\begin{align}
\partial_{t}f_{1}+\pi M\sqrt{\epsilon_{1}}f_{1}  & \geq\frac{8\pi^{2}}%
{\sqrt{2}}\int_{0}^{\infty}\int_{0}^{\infty}f_{3}f_{4}\left(  1+f_{1}%
+f_{2}\right)  Wd\epsilon_{3}d\epsilon_{4}\nonumber\\
& \geq\frac{8\pi^{2}}{\sqrt{2}}\int_{0}^{\infty}\int_{0}^{\infty}f_{3}%
f_{4}Wd\epsilon_{3}d\epsilon_{4}\ \label{H1E1}%
\end{align}
for $a.e.\ t\in\left[  T_{1},T_{2}\right]  $ and where $C_{0}>0.$

Since $f_{0}\not \equiv 0$ there exists an interval with the form $\left(
R_{1},\theta R_{1}\right)  $ with $R_{1}>0$ and $\theta>1$ perhaps close to
one, such that $\int_{R_{1}}^{\theta R_{1}}f_{0}d\epsilon>0.$ Moreover, due to
the continuity of the mild solutions in time in the weak topology in space it
follows that, assuming that $\bar{t}$ is small we have:%
\begin{equation}
\int_{R_{1}}^{\theta R_{1}}f\left(  t,\epsilon\right)  d\epsilon\geq
m_{1}>0\ \ ,\ \ 0\leq t\leq\bar{t}\label{H1E3}%
\end{equation}

Let us choose $\eta>0$ satisfying $\left(  2-\eta\right)  \theta
<2\ \ ,\ \left(  2-\eta\right)  >\theta>1.$ We can then derive a lower
estimate for $W $ if $\epsilon_{1}\in\left(  \left(  2-\eta\right)
R_{1},\left(  2-\eta\right)  \theta R_{1}\right)  .$ Combining (\ref{H1E1})
and (\ref{H1E3}) we obtain the estimate:
\begin{equation}
\partial_{t}f_{1}+\pi M\sqrt{\epsilon_{1}}f_{1}\geq K_{\eta,\theta}\left(
\int_{R_{1}}^{\theta R_{1}}fd\epsilon\right)  ^{2}\ \ ,\ \ \epsilon_{1}%
\in\left(  \theta R_{1},\left(  2-\eta\right)  \theta R_{1}\right)
\label{H1E3a}%
\end{equation}

Iterating this argument we obtain the chain of inequalities:%
\begin{equation}
\partial_{t}f_{1}+\pi M\sqrt{\epsilon_{1}}f_{1}\geq K_{\eta,\theta}\left(
\int_{R_{n}}^{\theta R_{n}}fd\epsilon\right)  ^{2}\ \ ,\ \ \epsilon_{1}%
\in\left(  \theta R_{n},\theta R_{n+1}\right) \label{H1E3b}%
\end{equation}
with:%
\begin{equation}
R_{n+1}=\left(  2-\eta\right)  R_{n}\ .\label{H1E3c}%
\end{equation}

Let us consider the sequence of positive times: $t_{n}=\bar t-\tau_{n}$ with
$\tau_{n}=\frac{b}{\sqrt{R_{n}}}$ with $0<b<1,$ independent of $n$. Using the
fact that $f\geq0$ we obtain, using Duhamel's formula:%
\begin{align}
& f\left(  t,\epsilon_{1}\right)  \geq K_{\eta,\theta}\int_{t_{n}}^{t}%
\exp\left(  -\pi M\sqrt{\epsilon_{1}}\left(  s-t_{n}\right)  \right)  \left(
\int_{R_{n}}^{\theta R_{n}}f\left(  s,\epsilon\right)  d\epsilon\right)
^{2}ds\nonumber\\
& \hbox{for}\,\,\,a.e.\,\,\epsilon_{1}\in\left(  \theta R_{n},\theta
R_{n+1}\right) \label{H1E3d}%
\end{align}
If we assume that $b$ is small, depending only on $\bar{t},$ we would obtain
that $t_{n}>0,$ for all $n\geq1.$ Using (\ref{H1E3d}) we then obtain:%
\begin{equation}
f\left(  t,\epsilon_{1}\right)  \geq K_{\eta,\theta}a\int_{t_{n}}^{t}\left(
\int_{R_{n}}^{\theta R_{n}}f\left(  s,\epsilon\right)  d\epsilon\right)
^{2}ds\ \ ,\ \ \epsilon_{1}\in\left(  \theta R_{n},\theta R_{n+1}\right)
\ \label{H1E3e}%
\end{equation}
for some $a>0$ independent on $n$. Applying (\ref{H1E3}) and (\ref{H1E3e})
with $n=1$ in the interval $\left(  R_{2},\theta R_{2}\right)  \subset\left(
\theta R_{1},\theta R_{2}\right)  ,$ we obtain:%
\begin{equation}
f\left(  t,\epsilon_{1}\right)  \geq K_{\eta,\theta}a\left(
t-t_{2}\right)  \left(  m_{1}\right)  ^{2}\ \ ,\ \ \epsilon_{1}\in\left(
R_{2},\theta R_{2}\right)  \ \ ,\ \ a.e.\text{\ }t\in\left[  t_{2},\bar
{t}\right] \label{H1E3g}%
\end{equation}

Integrating this inequality in $\left(  R_{2},\theta R_{2}\right)  $ we
obtain:%
\begin{equation}
\int_{R_{2}}^{\theta R_{2}}f\left(  t,\epsilon\right)  d\epsilon\geq
m_{2}=K_{\eta,\theta}a\left(  \theta-1\right)  R_{1}\left(  t-t_{2}\right)
\left(  m_{1}\right)  ^{2}\ \ ,\ \ a.e.\text{\ }t\in\left[  t_{2},\bar
{t}\right] \label{H1E3f}%
\end{equation}

Due to the definition of the sequence $\left\{  t_{n}\right\}  $ we have:%
\[
\left(  t_{3}-t_{2}\right)  =\left(  \bar{t}-\tau_{3}-t_{2}\right)  =\left(
\tau_{2}-\tau_{3}\right)  =\frac{b}{\sqrt{R_{2}}}\left(  1-\frac{1}%
{\sqrt{2-\eta}}\right)  \ge \frac{b\sqrt{2-\eta}}{\sqrt{R_{3}}}\left(  1-\frac
{1}{\sqrt{2}}\right)
\]

Using (\ref{H1E3f}) with $t=t_{3}$ it then follows that:%
\[
m_{2}\geq\bar{K}_{\eta,\theta,a,b}\frac{R_{1}}{\sqrt{R_{2}}} m_{1} ^{2}\geq
C_{\eta,\theta,a,b,R_{1}} m_{1}^{2}%
\]

Iterating the argument it then follows that:%
\[
m_{n+1}\geq C_{\eta,\theta,a,b,R_{1}}\left(  2-\eta\right)  ^{\frac{n-1}{2}%
}m_{n} ^{2}\ \ ,\ \ n=1,2,3,...
\]

Moreover, since the term $\left(  2-\eta\right)  ^{\frac{n}{2}}$ increases to
infinity we have that, after a finite number of steps the factor in front of
$m_{n} ^{2}$ becomes larger than one. We obtain lower estimates for $m_{n_{0}%
}$ in terms of the first term. Then:%
\[
m_{n+1}\geq m_{n} ^{2}\ \ ,\ \ n\geq n_{0}%
\]

Iterating we then obtain:%
\[
m_{n}\geq\left(  m_{n_{0}}\right)  ^{2^{n-n_{0}}}\ \ ,\ \ n\geq n_{0}%
\]

Since $m_{n_{0}}$ is a fixed number, in general smaller than one we obtain:%
\[
m_{n}\geq\exp\left(  -\tilde{\beta}_{1}2^{n-n_{0}}\right)
\]
with $\tilde{\beta}_{1}>0$ independent on $n.$ Since $R_{n}$ is of order
$\left(  2-\eta\right)  ^{n}$ this implies an estimate of the form:%
\[
m_{n}\geq C\exp\left(  -\tilde{\beta}R_{n}^{1+\omega}\right)
\ ,\ \ n=1,2,3,...
\]
for some $\tilde{\beta}>0\ ,\ C>0\, ,\omega>0$ independent on $n$, where
moreover $\omega$ is such that $1+\omega> \frac{log 2} {log (2-\eta)}$. Notice
that, since $\eta>0$ may be taken as small as we wish by taking $\theta$
sufficiently close to $1$, the value of $\omega$ may be as small as we need.

Using now (\ref{H1E3e}) as well as the fact that the union of the intervals
$\left(  R_{n},\theta R_{n}\right)  ,\ n\geq1$ cover the whole interval
$\left(  R_{1},\infty\right)  $ except a set of measure zero, we obtain:%
\begin{equation}
f\left(  t,\epsilon\right)  \geq C_{T_{1},T_{2},\omega}\exp\left(
-\epsilon^{1+\omega}\right)  \ \ ,\ \ \epsilon\geq R_{1}\ \ ,\ \ 0<T_{1}\leq
t\leq T_{2}\ \label{H1E4a}%
\end{equation}
where the value of $\omega$ may have been changed from  its previous
value in order to eliminate the constant $\tilde\beta$ from the exponential,
but still remaining as small as we need. The constant $C_{T_{1},T_{2},\omega}$
depends on $f_{0}.$

In order to obtain a lower bound for small values of $\varepsilon\in(0,
R_{1})$ we use again (\ref{H1E1}). By (\ref{H1E4a}), for any $\varepsilon
_{1}>0$ we have:
\begin{align}
\partial_{t}f_{1}+\pi M\sqrt{\epsilon_{1}}f_{1}  & \geq\frac{8\pi^{2}}%
{\sqrt{2}}\int_{0}^{\infty}\int_{0}^{\infty}f_{3}f_{4}Wd\epsilon_{3}%
d\epsilon_{4}\nonumber\\
& \geq \frac{8\pi^{2}}{\sqrt{2}}\int_{2R_{1}}^{\infty}\int_{0}^{\varepsilon
_{1} }f_{3}f_{4}Wd\epsilon_{3}d\epsilon_{4}\nonumber\\
& \geq \frac{C} {\sqrt{\varepsilon_{1}}}\int_{0}^{\varepsilon_{1}}%
\sqrt{\varepsilon_{3}}f_{3} d\varepsilon_{3}. \ \label{H1E1bis}%
\end{align}
Using now Proposition (\ref{lowMass}) we have, for all $\varepsilon_{1}\in(0,
\rho(E, M))$:
\begin{align}
\label{H1E1C}\int_{0}^{\varepsilon_{1}}\sqrt{\varepsilon_{3}}f_{3}
d\varepsilon_{3}\ge K\varepsilon_{1}^{3/2},\,\,\,\,\hbox{for}\,\,t\ge T_{0}.
\end{align}
From (\ref{H1E1bis}) and (\ref{H1E1C}) we deduce
\begin{align*}
\partial_{t}f_{1}+\pi M\sqrt{\epsilon_{1}}f_{1}  & \geq C \varepsilon_{1}
,\,\,\,\,\hbox{for}\,\,t\ge T_{0},
\end{align*}
and therefore, after integration in time from $T_{*}$ to $t$ for any
$T_{*}>T_{0}$:
\begin{align*}
f(t, \varepsilon_{1})\ge C\sqrt{\varepsilon_{1}}\left( 1-e^{-\pi M
\sqrt{\varepsilon_{1}}(t-T_{*})} \right)  ,\,\,\,\,\hbox{for}\,\,t> T_{*}.
\end{align*}
If moreover $(t-T_{*})<1$:
\begin{align*}
f(t, \varepsilon_{1})\ge C\varepsilon_{1} t\ge C(T_{1}, T_{2})\varepsilon_{1}.
\end{align*}
This shows that for some positive constant $C_{T_{1}, T_{2},\,\omega}$:
\begin{align*}
f(t, \varepsilon_{1})\ge C_{T_{1}, T_{2}, \omega}\,\varepsilon_{1},\,\,\,\,
\varepsilon_{1}\in(0, \rho), \,\,\, \,\,\frac{3T_{0}}{2}<T_{1}\le t < T_{2}.
\end{align*}
where $\rho$ is as in Proposition (\ref{lowMass}). Arguing as in the proof of
formula (6.15) above, we deduce after a finite number of iterations that for
some positive constant $C_{T_{1}, T_{2},\,\omega}$:
\begin{align}
\label{H1E1D}f(t, \varepsilon_{1})\ge C_{T_{1}, T_{2}, \omega}\,\varepsilon
_{1},\,\,\,\, \varepsilon_{1}\in(0, R_{1}), \,\,\, \,\,2T_{0}<T_{1}\le t <
T_{2},
\end{align}
where the constant $C_{T_{1},T_{2},\omega}$ might change from line to line.

We deduce then from (\ref{H1E4a}) and (\ref{H1E1D}):
\begin{align}
\label{H1E1E}f(t, \varepsilon_{1})\ge C_{T_{1}, T_{2}, \omega}\,\varepsilon
_{1}\exp\left(  -\epsilon^{1+\omega}\right)  ,\,\,\,\, \varepsilon_{1}>0,
\,\,\, \,\,T_{0}(E, M)<T_{1}\le t <T_{2}.
\end{align}

Using (\ref{H1E1E}) we obtain:%
\begin{equation}
\left\vert \log\left(  f\right)  \right\vert \leq C_{T_{1},T_{2},\omega
}\left(  1+|\log\varepsilon|+\epsilon^{1+\omega}\right)  \ \ ,\ \ \epsilon
\geq0\ \ ,\ \ 0<T_{1}\leq t\leq T_{2}.\label{H1E4}%
\end{equation}
We now use (\ref{H1E4}) to show that the different integral terms appearing in
the formula of the dissipation of the entropy (cf. (\ref{G2E6})) are finite.
Due to the boundedness of $f$ we just need to show that the following
integrals are finite:%
\begin{align*}
& \int\int\int f_{1}f_{2}\left(  \left\vert \log\left(  \frac{f_{1}}{1+f_{1}%
}\right)  \right\vert +\left\vert \log\left(  \frac{f_{2}}{1+f_{2}}\right)
\right\vert \right)  \Phi d\epsilon_{1}d\epsilon_{2}d\epsilon_{3}\\
& \int\int\int f_{1}f_{2}\left(  \left\vert \log\left(  \frac{f_{3}}{1+f_{3}%
}\right)  \right\vert +\left\vert \log\left(  \frac{f_{4}}{1+f_{4}}\right)
\right\vert \right)  \Phi d\epsilon_{1}d\epsilon_{2}d\epsilon_{3}\\
& \int\int\int f_{3}f_{4}\left(  \left\vert \log\left(  \frac{f_{1}}{1+f_{1}%
}\right)  \right\vert +\left\vert \log\left(  \frac{f_{2}}{1+f_{2}}\right)
\right\vert \right)  \Phi d\epsilon_{1}d\epsilon_{2}d\epsilon_{3}\\
& \int\int\int f_{3}f_{4}\left(  \left\vert \log\left(  \frac{f_{3}}{1+f_{3}%
}\right)  \right\vert +\left\vert \log\left(  \frac{f_{4}}{1+f_{4}}\right)
\right\vert \right)  \Phi d\epsilon_{1}d\epsilon_{2}d\epsilon_{3}%
\end{align*}
where $\epsilon_{4}=\epsilon_{1}+\epsilon_{2}-\epsilon_{3}.$ Replacing the
variable $\epsilon_{4}$ by the variable $\epsilon_{1}$ in the last two
integrals, and relabelling the number of the resulting integration variables
(namely $\epsilon_{2},\epsilon_{3},\epsilon_{4}$ to $\epsilon_{1},\epsilon
_{2},\epsilon_{3}$) we reduce the estimate of the last two integrals to the
first two ones. Using now the symmetry of the variables $\epsilon_{1}%
,\epsilon_{2}$ we are left only with the three different terms:%
\begin{align*}
I_{1}  & =\int\int\int f_{1}f_{2}\left\vert \log\left(  \frac{f_{1}}{1+f_{1}%
}\right)  \right\vert \Phi d\epsilon_{1}d\epsilon_{2}d\epsilon_{3}\\
I_{2}  & =\int\int\int f_{1}f_{2}\left\vert \log\left(  \frac{f_{3}}{1+f_{3}%
}\right)  \right\vert \Phi d\epsilon_{1}d\epsilon_{2}d\epsilon_{3},\\
\ I_{3}  & =\int\int\int f_{1}f_{2}\left\vert \log\left(  \frac{f_{4}}%
{1+f_{4}}\right)  \right\vert \Phi d\epsilon_{1}d\epsilon_{2}d\epsilon_{3}%
\end{align*}
Using (\ref{H1E4}) and the boundedness of $f$ we obtain:%
\begin{align*}
I_{1}  & \leq C\int\int\int f_{1}f_{2}\left(  1+\epsilon_{1}^{1+\omega
}\right)  \Phi d\epsilon_{1}d\epsilon_{2}d\epsilon_{3}\\
& =C\int\int\int_{\left\{  \epsilon_{1}\geq\epsilon_{2}\right\}  }\left[
\cdot\cdot\cdot\right]  \Phi d\epsilon_{1}d\epsilon_{2}d\epsilon_{3}+C\int
\int\int_{\left\{  \epsilon_{1}<\epsilon_{2}\right\}  }\left[  \cdot\cdot
\cdot\right]  \Phi d\epsilon_{1}d\epsilon_{2}d\epsilon_{3}%
\end{align*}

We now use that $\Phi\leq\sqrt{\epsilon_{2}}$ in the first integral and
$\Phi\leq\sqrt{\epsilon_{1}}$ in the second. Then, since $\epsilon_{3}%
\leq\epsilon_{1}+\epsilon_{2}$%
\begin{align*}
I_{1}  & \leq C\int\int_{\left\{  \epsilon_{1}\geq\epsilon_{2}\right\}  }%
f_{1}f_{2}\left(  1+\epsilon_{1}^{1+\omega}+|\log\varepsilon_{1}|\right)  \epsilon_{1}\sqrt
{\epsilon_{2}}d\epsilon_{1}d\epsilon_{2}+\\
& +C\int\int\int_{\left\{  \epsilon_{1}<\epsilon_{2}\right\}  }f_{1}%
f_{2}\left(  1+\epsilon_{1}^{1+\omega}+|\log\varepsilon_{1}|\right)  \epsilon_{2}\sqrt{\epsilon_{1}%
}d\epsilon_{1}d\epsilon_{2}%
\end{align*}

Since, at is has been indicated before, the value of $\omega$ in (\ref{H1E4})
may be chosen as small as we need, we will assume in the following that
$\omega<\left(  \gamma-3\right)  .$ Then, using that $f\leq\frac{C}{\left(
1+\epsilon\right)  ^{\gamma}}:$%
\begin{align*}
I_{1}  & \leq C\int\int_{\left\{  \epsilon_{1}\geq\epsilon_{2}\right\}
}\left(  1+\epsilon_{1}\right)  ^{-\gamma}\left(  1+\epsilon_{2}\right)
^{-\gamma}\left(  1+\epsilon_{1}^{1+\omega}+|\log\varepsilon_{1}|\right)  \epsilon_{1}\sqrt
{\epsilon_{2}}d\epsilon_{1}d\epsilon_{2}+\\
& +C\int\int\int_{\left\{  \epsilon_{1}<\epsilon_{2}\right\}  }\left(
1+\epsilon_{1}\right)  ^{-\gamma}\left(  1+\epsilon_{2}\right)  ^{-\gamma
}\left(  1+\epsilon_{1}^{1+\omega}+|\log\varepsilon_{1}|\right)  \epsilon_{2}\sqrt{\epsilon_{1}%
}d\epsilon_{1}d\epsilon_{2}\\
& \leq C
\end{align*}

In order to estimate $I_{2}$ and $I_{3}$ we use a symmetrization argument that
yields:%
\begin{align*}
I_{2}  & \leq C\int\int\int_{\left\{  \epsilon_{1}\geq\epsilon_{2}\right\}
}\left(  1+\epsilon_{1}\right)  ^{-\gamma}\left(  1+\epsilon_{2}\right)
^{-\gamma}\left(  1+\epsilon_{3}^{1+\omega}+|\log\varepsilon_{3}|\right)  \Phi
d\epsilon_{1}d\epsilon_{2}d\epsilon_{3}\\
& \leq C\int\int\int_{\left\{  \epsilon_{1}\geq\epsilon_{2}\right\}  }\left(
1+\epsilon_{1}\right)  ^{-\gamma}\left(  1+\epsilon_{2}\right)  ^{-\gamma
}\left(  1+\epsilon_{1}^{1+\omega}+|\log\varepsilon_{1}|\right)  \epsilon
_{1}\sqrt{\epsilon_{2}}d\epsilon_{1}d\epsilon_{2}\\
& \leq C
\end{align*}%
\begin{align*}
\ I_{3}  & \leq C\int\int\int_{\left\{  \epsilon_{1}\geq\epsilon_{2}\right\}
}\left(  1+\epsilon_{1}\right)  ^{-\gamma}\left(  1+\epsilon_{2}\right)
^{-\gamma}\left(  1+\epsilon_{4}^{1+\omega}+|\log\varepsilon_{4}|\right)  \Phi
d\epsilon_{1}d\epsilon_{2}d\epsilon_{3}\\
& \leq C\int\int_{\left\{  \epsilon_{1}\geq\epsilon_{2}\right\}  }\left(
1+\epsilon_{1}\right)  ^{-\gamma}\left(  1+\epsilon_{2}\right)  ^{-\gamma
}\left(  1+\epsilon_{1}^{1+\omega}+|\log\varepsilon_{1}|\right)  \sqrt
{\epsilon_{2}}d\epsilon_{1}d\epsilon_{2}\\
& \leq C
\end{align*}
where in the estimate of $I_{3}$ we have used that $\epsilon_{4}\leq
2\epsilon_{1}.$

In order to conclude the proof of (\ref{G4E7}) we need to use a symmetrization
argument. We will use that (cf. \cite{EV1}, Theorem 3.4):\
\[
\partial_{t}f_{1}=\frac{8\pi^{2}}{\sqrt{2}}\int\int_{D\left(  \epsilon
_{1}\right)  }W\left[  \left(  1+f_{1}\right)  \left(  1+f_{2}\right)
f_{3}f_{4}-\left(  1+f_{3}\right)  \left(  1+f_{4}\right)  f_{1}f_{2}\right]
d\epsilon_{3}d\epsilon_{4},\ a.e.\ t\in\left[  T_{1},T_{2}\right]
\]

Differentiating (\ref{G2E5}) and using the conservation of mass, we obtain the
following:%
\begin{align*}
\partial_{t}\left(  S\left[  f\right]  \right)   & =\int\int\int\log\left(
\frac{1+f_{1}}{f_{1}}\right)  \Phi\times\\
& \hskip0.6cm\times\left[  \left(  1+f_{1}\right)  \left(  1+f_{2}\right)
f_{3}f_{4}-\left(  1+f_{3}\right)  \left(  1+f_{4}\right)  f_{1}f_{2}\right]
d\epsilon_{1}d\epsilon_{3}d\epsilon_{4}\\
& =\int\int\int\log\left(  \frac{1+f_{1}}{f_{1}}\right)  \left(
1+f_{1}\right)  \left(  1+f_{2}\right)  \left(  1+f_{3}\right)  \left(
1+f_{4}\right)  \times\\
& \hskip5.4cm\times\left[  Q_{3,4}-Q_{1,2}\right]  \Phi d\epsilon_{1}%
d\epsilon_{3}d\epsilon_{4}%
\end{align*}

We now claim that for any function $f\in L^{\infty}\left(  \mathbb{R}%
_{+}:\left(  1+\epsilon\right)  ^{\gamma}\right)  $ satisfying (\ref{H1E4a})
the following identity holds:%
\begin{equation}
J_{1}=J_{2}\ \label{H1E4c}%
\end{equation}
where:%
\begin{align*}
J_{1}  & =\int\int\int\log\left(  \frac{1+f_{1}}{f_{1}}\right)  \left(
1+f_{1}\right)  \left(  1+f_{2}\right)  \left(  1+f_{3}\right)  \left(
1+f_{4}\right)  \times\\
& \hskip5.4cm\times\left[  Q_{3,4}-Q_{1,2}\right]  \Phi d\epsilon_{1}%
d\epsilon_{3}d\epsilon_{4}%
\end{align*}%
\begin{align*}
J_{2}  & =\frac{1}{4}\int\int\int\int\left[  \log\left(  Q_{3,4}\right)
-\log\left(  Q_{1,2}\right)  \right]  \left[  Q_{3,4}-Q_{1,2}\right]  \left(
1+f_{1}\right)  \left(  1+f_{2}\right)  \times\\
& \hskip6.5cm\times\left(  1+f_{3}\right)  \left(  1+f_{4}\right)  \Phi
d\epsilon_{1}d\epsilon_{2}d\epsilon_{3}%
\end{align*}

In order to prove (\ref{H1E4c}) suppose first that $f\in L^{\infty}\left(
\mathbb{R}_{+}:\left(  1+\epsilon\right)  ^{\gamma}\right)  \cap C\left(
\mathbb{R}_{+}\right)  $ satisfies (\ref{H1E4a}). Then:%
\begin{align*}
J_{1}  & =\frac{1}{4}\int\int\int\int\left[  \log\left(  Q_{3,4}\right)
-\log\left(  Q_{1,2}\right)  \right]  \left[  Q_{3,4}-Q_{1,2}\right]  \left(
1+f_{1}\right)  \left(  1+f_{2}\right)  \times\\
& \hskip3cm\times\left(  1+f_{3}\right)  \left(  1+f_{4}\right)  \Phi
\delta\left(  \epsilon_{1}+\epsilon_{2}-\epsilon_{3}-\epsilon_{4}\right)
d\epsilon_{1}d\epsilon_{2}d\epsilon_{3}d\epsilon_{4}\\
& =J_{2}%
\end{align*}
whence (\ref{H1E4c}) holds for continuous functions. For arbitrary functions
$f\in L^{\infty}\left(  \mathbb{R}_{+}:\left(  1+\epsilon\right)  ^{\gamma
}\right)  $ satisfying (\ref{H1E4a}) we can obtain (\ref{H1E4c}) approximating
$f$ by means of a sequence of continuous function $f_{n}$\ converging to $f$
at almost every $t\in\left[  T_{1},T_{2}\right]  .$
\end{proof}

\section{Reformulation of the criticality condition: A technical Lemma.}

\setcounter{equation}{0} \setcounter{theorem}{0}

We now prove an auxiliary result that reformulates the conditions in
Proposition \ref{crit} in a form that only depends on the values of the
equilibrium distributions for values of $\epsilon\geq R>0,$ with $R$ small. We
define a class of auxiliary functions:%
\begin{equation}
f_{s}\left(  \epsilon;\alpha,\beta\right)  =\frac{1}{\exp\left(  \beta\left(
\epsilon+\alpha\right)  \right)  -1}\ \ ,\ \ \epsilon>-\alpha\ \ ,\ \ \beta
>0,\ \ \alpha\in\mathbb{R}\label{frad}%
\end{equation}

The following result holds:

\begin{proposition}
\label{MassCrit}Given $E_{\ast}>0,$ and $\delta>0$ there exist $R_{0}\left(
E_{\ast},\delta\right)  >0,\ L_{0}\left(  E_{\ast},\delta\right)  >0$ such
that, if $f_{s}$ is one of the functions in (\ref{frad}), $0\leq R\leq
R_{0}\left( E_{\ast}, \delta\right)  ,\ L\geq L_{0}\left( E_{\ast},
\delta\right)  $, $\alpha\geq-\frac{R}{2}$ and
\[
E=4\pi\int_{R}^{L}f_{s}\left(  \epsilon;\alpha,\beta\right)  \sqrt
{2\epsilon^{3}}d\epsilon\ \ ,\ \ E\leq E_{\ast}%
\]
then:%
\begin{equation}
4\pi\int_{R}^{L}f_{s}\left(  \epsilon;\alpha,\beta\right)  \sqrt{2\epsilon
}d\epsilon\leq\frac{\zeta\left(  \frac{3}{2}\right)  }{\left(  \zeta\left(
\frac{5}{2}\right)  \right)  ^{\frac{3}{5}}}\left(  \frac{4\pi}{3}\right)
^{\frac{3}{5}}E^{\frac{3}{5}}+\delta\label{E2E6}%
\end{equation}

\end{proposition}

\begin{proof}
We will assume in all the following that $R\leq1,\ L\geq2.$ Suppose that
$E\leq E_{\ast}$. We define a family of functions function $\beta
_{E,R,L}:\left[  -\frac{R}{2},\infty\right)  \rightarrow\mathbb{R}^{+}$ by
means of the relation:
\[
E=4\pi\int_{R}^{L}f_{s}\left(  \epsilon;\alpha,\beta_{E,R,L}\left(
\alpha\right)  \right)  \sqrt{2\epsilon^{3}}d\epsilon
\]

Since the functions $f_{s}\left(  \epsilon;\alpha,\beta\right)  $ are strictly
decreasing with respect to $\beta$ and, for any given $\alpha\geq-\frac{R}{2}$
we have:%
\begin{align*}
4\pi\lim_{\beta\rightarrow0}\int_{R}^{L}f_{s}\left(  \epsilon;\alpha
,\beta\right)  \sqrt{2\epsilon^{3}}d\epsilon & =\infty\\
4\pi\lim_{\beta\rightarrow\infty}\int_{R}^{L}f_{s}\left(  \epsilon
;\alpha,\beta\right)  \sqrt{2\epsilon^{3}}d\epsilon & =0
\end{align*}
it follows that the function $\beta_{E,R,L}$ is well defined. We now define
the functions $M_{E,R;L}:\left[  -\frac{R}{2},\infty\right)  \rightarrow
\mathbb{R}^{+}$ as:%
\[
M_{E,R,L}\left(  \alpha\right)  =4\pi\int_{R}^{L}f_{s}\left(  \epsilon
;\alpha,\beta_{E,R,L}\left(  \alpha\right)  \right)  \sqrt{2\epsilon}d\epsilon
\]

We now claim that
\begin{equation}
\frac{dM_{E,R,L}\left(  \alpha\right)  }{d\alpha}<0\label{E2E2}%
\end{equation}
for $\alpha\in\left[  -\frac{R}{2},\infty\right)  .$ Indeed, differentiating
$M_{E,R,L}$ we obtain $\frac{dM_{E,R,L}\left(  \alpha\right)  }{d\alpha}%
=\frac{\Delta_{1}\left(  \alpha\right)  }{\Delta_{2}\left(  \alpha\right)  },$
where:%
\begin{align*}
\Delta_{1}\left(  \alpha\right)   & =8\pi\int_{R}^{L}dx\int_{R}^{L}%
dy\frac{e^{\beta\left(  x+\alpha\right)  }e^{\beta\left(  y+\alpha\right)  }%
}{\left(  e^{\beta\left(  x+\alpha\right)  }-1\right)  ^{2}\left(
e^{\beta\left(  y+\alpha\right)  }-1\right)  ^{2}}\times\\
& \times\left[  \sqrt{xy^{3}}\left(  x+\alpha\right)  -\sqrt{yx^{3}}\left(
x+\alpha\right)  \right]
\end{align*}%
\[
\Delta_{2}\left(  \alpha\right)  =\int_{R}^{L}dx\frac{e^{\beta\left(
x+\alpha\right)  }\sqrt{2x^{3}}\left(  x+\alpha\right)  }{\left(
e^{\beta\left(  x+\alpha\right)  }-1\right)  ^{2}}>0
\]

Symmetrizing the variables $x,y$ in $\Delta_{1}\left(  \alpha\right)  $ we
obtain:%
\[
\Delta_{1}\left(  \alpha\right)  =-8\pi\int_{R}^{L}dx\int_{R}^{L}%
dy\frac{e^{\beta\left(  x+\alpha\right)  }e^{\beta\left(  y+\alpha\right)
}\sqrt{xy}\left(  x-y\right)  ^{2}}{\left(  e^{\beta\left(  x+\alpha\right)
}-1\right)  ^{2}\left(  e^{\beta\left(  y+\alpha\right)  }-1\right)  ^{2}}<0
\]
whence (\ref{E2E2}) follows. Therefore:%
\begin{equation}
M_{E,R,L}\left(  \alpha\right)  \leq M_{E,R,L}\left(  -\frac{R}{2}\right)
=4\pi\int_{R}^{\infty}f_{s}\left(  \epsilon;-\frac{R}{2},\beta_{E,R,L}\left(
-\frac{R}{2}\right)  \right)  \sqrt{2\epsilon}d\epsilon\label{G8E1}%
\end{equation}
where $\beta_{E,R,L}\left(  -\frac{R}{2}\right)  $ satisfies:%
\begin{equation}
E=4\pi\int_{R}^{L}f_{s}\left(  \epsilon;-\frac{R}{2},\beta_{E,R,L}\left(
-\frac{R}{2}\right)  \right)  \sqrt{2\epsilon^{3}}d\epsilon\label{G8E2}%
\end{equation}

Then, since $R\leq1,\ L\geq2:$%
\begin{equation}
E\geq4\pi\int_{1}^{2}f_{s}\left(  \epsilon;-\frac{R}{2},\beta_{E,R,L}\left(
-\frac{R}{2}\right)  \right)  \sqrt{2\epsilon^{3}}d\epsilon>\frac{4\pi
}{e^{2\beta_{E,R,L}\left(  -\frac{R}{2}\right) } }\int_{1}^{2}\frac
{\sqrt{2\epsilon^{3}}d\epsilon}{\epsilon}\label{E4E7}%
\end{equation}
It then follows from the fact that $E\leq E_{\ast}$ that there exists
$\beta_{\ast}=\beta_{\ast}\left(  E_{\ast}\right)  >0$ such that:%
\begin{equation}
\beta_{E,R,L}\left(  -\frac{R}{2}\right)  \geq\beta_{\ast}\label{E2E3}%
\end{equation}

Using (\ref{E2E3}) it follows that, for any $\varepsilon_{0}>0,$ there exists
$\bar{L}_{0}=\bar{L}_{0}\left(  E_{\ast},\varepsilon_{0}\right)  >0$ such
that, if $L\geq\bar{L}_{0}:$%
\begin{equation}
4\pi\int_{L}^{\infty}f_{s}\left(  \epsilon;-\frac{R}{2},\beta_{E,R,L}\left(
-\frac{R}{2}\right)  \right)  \sqrt{2\epsilon^{3}}d\epsilon<\varepsilon
_{0}\label{E2E4}%
\end{equation}

We define the functions:%
\begin{equation}
\Phi_{R}\left(  \beta\right)  =4\pi\int_{R}^{\infty}f_{s}\left(
\epsilon;-\frac{R}{2},\beta\right)  \sqrt{2\epsilon^{3}}d\epsilon\label{E2E5}%
\end{equation}

We define $\tilde{E}$ by means of:%
\[
\tilde{E}=E+4\pi\int_{L}^{\infty}f_{s}\left(  \epsilon;-\frac{R}{2}%
,\beta_{E,R,L}\left(  -\frac{R}{2}\right)  \right)  \sqrt{2\epsilon^{3}%
}d\epsilon
\]

Using (\ref{E2E4}) we obtain $E\leq\tilde{E}\leq E+\varepsilon_{0}.$ Notice
that due to (\ref{G8E2}) we have $\Phi_{R}\left(  \beta_{E,R,L}\left(
-\frac{R}{2}\right)  \right)  =\tilde{E}.$

We now claim that there exists $\varepsilon_{1}>0$ depending only on $E_{\ast
},\ \delta$ such that if $E\leq\varepsilon_{1}$ and $\alpha\geq-\frac{R}{2}$
we have $M_{E,R,L}\left(  \alpha\right)  \leq\delta.$ Indeed, (\ref{E4E7})
implies that $\beta_{E,R,L}\left(  -\frac{R}{2}\right)  $ can be made
arbitrarily large if $\varepsilon_{1}$ is small. Then (\ref{frad}) and
(\ref{G8E1}) implies that $M_{E,R,L}\left(  \alpha\right)  $ can be made small
due to Lebesgue's dominated convergence Theorem. Then (\ref{E2E6}) would
follow in this case. We remark that $\varepsilon_{1}$ is independent on $R.$

We will assume then in the following that $E>\varepsilon_{1}.$ We claim that
there exists $\beta^{\ast}=\beta^{\ast}\left(  E_{\ast},\varepsilon
_{1}\right)  ,$ such that $\beta_{E,R,L}\left(  -\frac{R}{2}\right)  \leq
\beta^{\ast}.$ Indeed, we can estimate $\Phi_{R}\left(  \beta\right)  $ as:%
\[
\Phi_{R}\left(  \beta\right)  \leq4\pi\int_{\frac{R}{2}}^{\infty}\frac
{\sqrt{2\left(  \epsilon+\frac{R}{2}\right)  ^{3}}}{\left(  e^{\beta\epsilon
}-1\right)  }d\epsilon\leq8\pi\int_{0}^{\infty}\frac{\sqrt{\epsilon^{3}}%
}{\left(  e^{\beta\epsilon}-1\right)  }d\epsilon
\]
and the right-hand side of this formula converges to zero as $\beta
\rightarrow\infty$ due to Lebesgue's dominated convergence Theorem. Then,
$\beta_{E,R,L}\left(  -\frac{R}{2}\right)  \leq\beta^{\ast}=\beta^{\ast
}\left(  E_{\ast},\varepsilon_{1}\right)  $ if $E\geq\varepsilon_{1}>0.$ We
now claim that the functions $\Phi_{R}\left(  \beta\right)  $ converge
uniformly\ in the interval $\left[  \beta_{\ast},\beta^{\ast}\right]  $ as
$R\rightarrow0$ to the function:%
\begin{equation}
\Phi_{0}\left(  \beta\right)  =4\pi\int_{0}^{\infty}f_{s}\left(
\epsilon;0,\beta\right)  \sqrt{2\epsilon^{3}}d\epsilon=4\pi\int_{0}^{\infty
}\frac{\sqrt{2\epsilon^{3}}}{\left(  e^{\beta\epsilon}-1\right)  }%
d\epsilon\label{E3E2}%
\end{equation}

Moreover, we have also uniform convergence in the interval $\left[
\beta_{\ast},\beta^{\ast}\right]  $ of the derivatives $\Phi_{R}^{\prime
}\left(  \beta\right)  $ to $\Phi_{0}^{\prime}\left(  \beta\right)  $ as
$R\rightarrow0.$

Indeed, this just follows from the inequalities:
\[
\left\vert \Phi_{R}\left(  \beta\right)  -\Phi_{0}\left(  \beta\right)
\right\vert \leq C\int_{\frac{R}{2}}^{\infty}\frac{\left[  \sqrt{\left(
\epsilon+\frac{R}{2}\right)  ^{3}}-\sqrt{\epsilon^{3}}\right]  }{\epsilon
}e^{-\beta_{\ast}\epsilon}d\epsilon+C\int_{0}^{\frac{R}{2}}\sqrt{\epsilon
}d\epsilon
\]%
\[
\left\vert \Phi_{R}^{\prime}\left(  \beta\right)  -\Phi_{0}^{\prime}\left(
\beta\right)  \right\vert \leq C\int_{\frac{R}{2}}^{\infty}\frac{\left[
\sqrt{\left(  \epsilon+\frac{R}{2}\right)  ^{3}}-\sqrt{\epsilon^{3}}\right]
}{\epsilon}e^{-\beta_{\ast}\epsilon}d\epsilon+C\int_{0}^{\frac{R}{2}}%
\sqrt{\epsilon}d\epsilon
\]

Since $\Phi_{0}^{\prime}\left(  \beta\right)  $ is strictly negative for
$\beta\in\left[  \beta_{\ast},\beta^{\ast}\right]  $ we then obtain that, if
$E\geq\varepsilon_{1},$ $\beta_{E,R,L}\left(  -\frac{R}{2}\right)  $
satisfies:%
\begin{equation}
\left\vert \beta_{E,R,L}\left(  -\frac{R}{2}\right)  -\bar{\beta}\left(
E\right)  \right\vert \leq C\varepsilon_{0}\ \label{G8E3}%
\end{equation}
if $R$ is sufficiently small, where $\bar{\beta}\left(  E\right)  $ is the
unique solution of:%
\[
\Phi_{0}\left(  \bar{\beta}\left(  E\right)  \right)  =E
\]

Similar computations yield the uniform convergence in the interval $\beta
\in\left[  \beta_{\ast},\beta^{\ast}\right]  $ of the functions $\Psi
_{R}\left(  \beta\right)  $ defined by means of:%
\[
\Psi_{R}\left(  \beta\right)  =4\pi\int_{R}^{\infty}f_{s}\left(
\epsilon;-\frac{R}{2},\beta\right)  \sqrt{2\epsilon}d\epsilon
\]
to:%
\begin{equation}
\Psi_{0}\left(  \beta\right)  =4\pi\int_{0}^{\infty}\frac{\sqrt{2\epsilon}%
}{\left(  e^{\beta\epsilon}-1\right)  }d\epsilon\ \label{E3E1}%
\end{equation}
if $E\geq\varepsilon_{1}.$ Using (\ref{G8E1}) we can write:%
\begin{equation}
M_{E,R,L}\left(  -\frac{R}{2}\right)  =\left[  \Psi_{R}\left(  \beta
_{E,R,L}\left(  -\frac{R}{2}\right)  \right)  -\Psi_{0}\left(  \beta
_{E,R,L}\left(  -\frac{R}{2}\right)  \right)  \right]  +\Psi_{0}\left(
\beta_{E,R,L}\left(  -\frac{R}{2}\right)  \right) \label{E2E8}%
\end{equation}

Due to (\ref{G8E3}) and the uniform convergence of the functions $\Psi_{R}$ to
$\Psi_{0}$ stated above, it follows that the term between brackets in
(\ref{E2E8}) can be made arbitrarily small if $R$ is small enough. On the
other hand, using again (\ref{G8E3}) we can make $\Psi_{0}\left(
\beta_{E,R,L}\left(  -\frac{R}{2}\right)  \right)  $ arbitrarily close to
$\Psi_{0}\left(  \bar{\beta}\left(  E\right)  \right)  $ if $\varepsilon_{0}$
is small. Therefore, if $R$ is small we obtain:%
\begin{equation}
M_{E,R,L}\left(  -\frac{R}{2}\right)  \leq\Psi_{0}\left(  \bar{\beta}\left(
E\right)  \right)  +\delta\label{G8E4}%
\end{equation}
Using the definitions of $\Phi_{0}$ and $\Psi_{0}$ in (\ref{E3E2}),
(\ref{E3E1}) as well as Proposition \ref{crit} we have:%
\[
\Psi_{0}\left(  \bar{\beta}\left(  E\right)  \right)  \leq\frac{\zeta\left(
\frac{3}{2}\right)  }{\left(  \zeta\left(  \frac{5}{2}\right)  \right)
^{\frac{3}{5}}}\left(  \frac{4\pi}{3}\right)  ^{\frac{3}{5}}\left(  \Phi
_{0}\left(  \bar{\beta}\left(  E\right)  \right)  \right)  ^{\frac{3}{5}%
}=\frac{\zeta\left(  \frac{3}{2}\right)  }{\left(  \zeta\left(  \frac{5}%
{2}\right)  \right)  ^{\frac{3}{5}}}\left(  \frac{4\pi}{3}\right)  ^{\frac
{3}{5}}\left(  E\right)  ^{\frac{3}{5}}%
\]

Combining this estimate with (\ref{G8E4}) the result follows.
\end{proof}

\section{Mass concentration in the region of small energies for a sequence of
times $t_{n}\to+\infty$}

\setcounter{equation}{0} \setcounter{theorem}{0}

We will prove now that if the solutions of (\ref{F3E2}) are globally bounded
and (\ref{G1E2}) holds, the corresponding functions $g\left(  t,\cdot\right)
$ would have a significant amount of mass in the regions where $\epsilon$ is
small. The main result of this Section is the following.

\begin{proposition}
\label{ConcMass}Suppose that $f_{0},f$ are as in Theorem \ref{Main}. Let us
assume that $T_{\max}=\infty.$ Then, there exists $m_{\ast}>0$ and $\rho>0,$
both of them depending only on $M,\ E$ such that, for any $0<R<\rho$\ there
exists a sequence $\left\{  t_{n}\right\}  $ with $t_{n}\rightarrow\infty$ as
$n\rightarrow\infty$ such that:%
\begin{equation}
\int_{0}^{R}g\left(  t_{n},\epsilon\right)  d\epsilon=4\pi\int_{0}^{R}%
\sqrt{2\epsilon}f\left(  t_{n},\epsilon\right)  d\epsilon\geq m_{\ast
}\label{G4E1}%
\end{equation}
for any $n.$
\end{proposition}

In order to prove Proposition \ref{ConcMass} we need several Lemmas. We begin
deriving an estimate for the number of particles with large energy.

\begin{lemma}
\label{LemmaL}Suppose that $f_{0}$ and $f$ are as in Theorem \ref{Main}. Then,
for any $\varepsilon_{0}>0,$ there exists $L=L\left(  E,\varepsilon
_{0}\right)  $ such that%
\[
\int_{0}^{L}g\left(  t,\epsilon\right)  d\epsilon=4\pi\int_{0}^{L}%
\sqrt{2\epsilon}f\left(  t,\epsilon\right)  d\epsilon\geq M-\varepsilon_{0}%
\]
for $t\in\left[  0,T_{\max}\right]  .$
\end{lemma}

\begin{proof}
It is just a consequence from the conservation of energy $E$ as well as the
inequality:%
\[
\int_{L}^{\infty}g\left(  t,\epsilon\right)  d\epsilon\leq\frac{1}{L}\int
_{L}^{\infty}g\left(  t,\epsilon\right)  \epsilon d\epsilon \le \frac{E}{L}%
\]

Choosing $L\geq\frac{E}{\varepsilon_{0}}$ the result follows.
\end{proof}

We define the following auxiliary function:
\[
Q\left(  t,\epsilon\right)  =\frac{f\left(  t,\epsilon\right)  }{1+f\left(
t,\epsilon\right)  }%
\]

Then:%
\begin{equation}
f\left(  t,\epsilon\right)  =\frac{Q\left(  t,\epsilon\right)  }{1-Q\left(
t,\epsilon\right)  }\label{G4E2a}%
\end{equation}

Notice that:%
\begin{equation}
0\leq Q \left(  t,\epsilon\right)  \leq1\ \ \ \ ,\ \ Q\left(  t,\epsilon
\right)  \leq f\left(  t,\epsilon\right)  \ \ ,\ \ \epsilon\geq0,\ \ t\in
\left[  0,T_{\max}\right] \label{G4E2}%
\end{equation}

Then:%
\[
4\pi\int_{L}^{\infty}\sqrt{2\epsilon}f\left(  t,\epsilon\right)  d\epsilon
\leq\varepsilon_{0}\ \ ,
\]
with $L=L\left(  E,\varepsilon_{0}\right)  $ as in Lemma \ref{LemmaL}.
Moreover
\begin{equation}
4\pi\int_{0}^{\infty}\sqrt{2\epsilon}Q\left(  t,\epsilon\right)  d\epsilon\leq
M\ \ ,\ \ 4\pi\int_{0}^{\infty}\sqrt{2\epsilon^{3}}Q\left(  t,\epsilon\right)
d\epsilon\leq E\label{G4E3}%
\end{equation}

We define also:%
\begin{equation}
\Psi\left(  s\right)  =s\log\left(  1+s\right) \label{G4E4}%
\end{equation}

We will use the following concept of weak convergence.

\begin{definition}
We say that a sequence $\left\{  Q_{n}\right\}  \subset L^{\infty}\left(
\mathbb{R}^{+}\right)  $ converges weakly to $Q_{\ast},$ and we will write in
this case $Q_{n}\rightharpoonup Q_{\ast}$ iff:%
\begin{equation}
\lim_{n\rightarrow\infty}\int_{\mathbb{R}^{+}}Q_{n}\varphi d\epsilon
=\int_{\mathbb{R}^{+}}Q_{\ast}\varphi d\epsilon\ \label{G4E4a}%
\end{equation}
for any test function $\varphi\in C_{0}\left[  0,\infty\right)  .$
\end{definition}

\begin{remark}
\label{RemTest}If the sequence of functions $\left\{  Q_{n}\right\}  $
satisfies $0\leq Q_{n}\leq1,$ we can apply a density argument to show that
(\ref{G4E4a}) holds for any $\varphi\in L^{1}\left(  0,\infty\right)  .$
\end{remark}

We have the following result:

\begin{lemma}
\label{LemmaSubSeq}Suppose that $f_{0},f$ are as in Proposition \ref{ConcMass}%
. There exists a sequence $\left\{  t_{n}\right\}  ,$ $t_{n}\rightarrow\infty$
as $n\rightarrow\infty$ such that:%
\begin{equation}
\int_{\mathbb{R}^{+}}\int_{\mathbb{R}^{+}}\int_{\mathbb{R}^{+}}f\left(
t_{n},\epsilon_{1}\right)  f\left(  t_{n},\epsilon_{2}\right)  \Psi\left(
\frac{Q\left(  t_{n},\epsilon_{3}\right)  Q\left(  t_{n},\epsilon_{4}\right)
}{Q\left(  t_{n},\epsilon_{1}\right)  Q\left(  t_{n},\epsilon_{2}\right)
}-1\right)  \Phi d\epsilon_{1}d\epsilon_{2}d\epsilon_{3}\rightarrow
0\label{G4E5}%
\end{equation}
as $n\to\infty$, where $\epsilon_{4}=\epsilon_{1}+\epsilon_{2}-\epsilon_{3}$.
Moreover:%
\begin{equation}
Q\left(  t_{n},\cdot\right)  \rightharpoonup Q_{\ast}\left(  \cdot\right)
\ \ \text{as\ \ }n\rightarrow\infty\ \label{G4E6}%
\end{equation}
where $Q_{\ast}\in L^{\infty}\left(  \mathbb{R}^{+}\right)  $ and $0\leq
Q_{\ast}\left(  \epsilon\right)  \leq1\ \ ,\ \ \epsilon\geq0.$
\end{lemma}

\begin{proof}
Notice that (\ref{G2E6}) yields:%
\begin{equation}
D\left[  f\right]  \geq\int_{\mathbb{R}^{+}}\int_{\mathbb{R}^{+}}%
\int_{\mathbb{R}^{+}}f\left(  \epsilon_{1}\right)  f\left(  \epsilon
_{2}\right)  \Psi\left(  \frac{Q\left(  \epsilon_{3}\right)  Q\left(
\epsilon_{4}\right)  }{Q\left(  \epsilon_{1}\right)  Q\left(  \epsilon
_{2}\right)  }-1\right)  \Phi d\epsilon_{1}d\epsilon_{2}d\epsilon
_{3}\label{G4E8}%
\end{equation}

Due to (\ref{G4E7a}) we have $\left\vert S\left[  f\right]  \left(
T_{2}\right)  \right\vert +\left\vert S\left[  f\right]  \left(  T_{1}\right)
\right\vert \leq2C\left(  E,M\right)  .$\ Then, since $T_{\max}=\infty,$ and
$\left\vert S\left[  f\right]  \left(  T_{2}\right)  \right\vert <\infty,$ we
can use (\ref{G4E7}) to obtain:%
\[
\int_{T_{1}}^{\infty}D\left[  f\left(  \cdot,t\right)  \right]  dt<\infty
\]

Therefore, there exists a sequence $\left\{  t_{n}\right\}  ,$ $t_{n}%
\rightarrow\infty$ such that $D\left[  f\left(  \cdot,t_{n}\right)  \right]
\rightarrow0$ as $n\rightarrow\infty.$ Using then (\ref{G4E8}) we obtain
(\ref{G4E5}). Using (\ref{G4E2}) and classical compactness results for
measures in the weak topology, we can then extract a subsequence of $\left\{
t_{n}\right\}  $ (which will be denoted in the same way) for which
(\ref{G4E6}) holds.
\end{proof}

\begin{lemma}
\label{R1R2}Suppose that $f_{0},f$ are as in Lemma \ref{LemmaSubSeq}. Let us
assume that $Q_{\ast}=0$ in (\ref{G4E6}) for $0<R_{1}\leq\epsilon\leq
R_{2}<\infty$\thinspace. Then:%
\begin{equation}
\int_{R_{1}}^{R_{2}}Q\left(  t_{n},\epsilon\right)  \sqrt{\epsilon}%
d\epsilon\rightarrow0\text{ }\ \ ,\ \ \int_{R_{1}}^{R_{2}}f\left(
t_{n},\epsilon\right)  \sqrt{\epsilon}d\epsilon\rightarrow0\ \ \text{as
}n\rightarrow\infty\label{G4E9}%
\end{equation}

\end{lemma}

\begin{proof}
Choosing the test function $\varphi$ such that $\varphi\left(  \epsilon
\right)  =1$ for $\epsilon\in\left[  R_{1},R_{2}\right]  $ and $\varphi\left(
\epsilon\right)  =0$ for $\epsilon\notin\left[  R_{1},R_{2}\right]  $ (cf.
Remark \ref{RemTest}), we obtain, using the nonnegativity of $Q,$ that for $n$
large enough:%
\[
\int_{R_{1}}^{R_{2}}Q\left(  t_{n},\epsilon\right)  \sqrt{\epsilon}%
d\epsilon\rightarrow0\ \ \text{as }n\rightarrow\infty
\]

In order to prove the second formula in (\ref{G4E9}) we define the set:%
\[
\mathcal{U}_{n}\left(  R_{1},R_{2}\right)  =\left\{  R_{1}\leq\epsilon\leq
R_{2}:Q\left(  t_{n},\epsilon\right)  \geq\frac{1}{2}\right\}
\]

Then:%
\begin{align*}
& \int_{\mathbb{R}^{+}}\int_{\mathbb{R}^{+}}\int_{\mathbb{R}^{+}}f\left(
t_{n},\epsilon_{1}\right)  f\left(  t_{n},\epsilon_{2}\right)  \Psi\left(
\frac{Q\left(  t_{n},\epsilon_{3}\right)  Q\left(  t_{n},\epsilon_{4}\right)
}{Q\left(  t_{n},\epsilon_{1}\right)  Q\left(  t_{n},\epsilon_{2}\right)
}-1\right)  \Phi d\epsilon_{1}d\epsilon_{2}d\epsilon_{3}\\
& \geq\int_{\mathcal{U}_{n}\left(  R_{1},R_{2}\right)  }\int_{\mathcal{U}%
_{n}\left(  R_{1},R_{2}\right)  }\int_{\left\{  Q\left(  t_{n},\epsilon
_{3}\right)  \leq\frac{1}{8}\right\}  }f\left(  t_{n},\epsilon_{1}\right)
f\left(  t_{n},\epsilon_{2}\right)  \times\\
& \hskip4cm\times\Psi\left(  \frac{Q\left(  t_{n},\epsilon_{3}\right)
Q\left(  t_{n},\epsilon_{4}\right)  }{Q\left(  t_{n},\epsilon_{1}\right)
Q\left(  t_{n},\epsilon_{2}\right)  }-1\right)  \Phi d\epsilon_{3}
d\epsilon_{1}d\epsilon_{2}%
\end{align*}

Using that $Q\left(  t_{n},\epsilon_{4}\right)  \leq1,$ $Q\left(
t_{n},\epsilon_{j}\right)  \geq\frac{1}{2}$ if $\epsilon_{j}\in\mathcal{U}%
_{n}\left(  R_{1},R_{2}\right)  $ for some $j=1,2,$ as well as the fact that
the function $\Psi\left(  s\right)  $ is decreasing for $s\in\left(
-1,-\frac{1}{2}\right)  $ we obtain:%
\begin{align*}
& \int_{\mathbb{R}^{+}}\int_{\mathbb{R}^{+}}\int_{\mathbb{R}^{+}}f\left(
t_{n},\epsilon_{1}\right)  f\left(  t_{n},\epsilon_{2}\right)  \Psi\left(
\frac{Q\left(  t_{n},\epsilon_{3}\right)  Q\left(  t_{n},\epsilon_{4}\right)
}{Q\left(  t_{n},\epsilon_{1}\right)  Q\left(  t_{n},\epsilon_{2}\right)
}-1\right)  \Phi d\epsilon_{1}d\epsilon_{2}d\epsilon_{3}\\
& \geq\int_{\mathcal{U}_{n}\left(  R_{1},R_{2}\right)  }\int_{\mathcal{U}%
_{n}\left(  R_{1},R_{2}\right)  }\int_{\left\{  Q\left(  t_{n},\epsilon
_{3}\right)  \leq\frac{1}{8}\right\}  \cap\left\{  \frac{R_{1}}{2}\leq
\epsilon_{3}\leq\frac{3R_{1}}{2}\right\}  }f\left(  t_{n},\epsilon_{1}\right)
f\left(  t_{n},\epsilon_{2}\right) \times\\
&  \hskip 6cm \times\Psi\left(  4Q\left(  t_{n},\epsilon_{3}\right)
-1\right)  \Phi d\epsilon_{3} d\epsilon_{1}d\epsilon_{2}\\
& \geq\sqrt{\frac{R_{1}}{2}}\left(  \int_{\mathcal{U}_{n}\left(  R_{1}%
,R_{2}\right)  }f\left(  t_{n},\epsilon\right)  d\epsilon\right)  ^{2}\times\\
& \hskip 3cm \times\left[  \int_{\left\{  Q\left(  t_{n},\epsilon_{3}\right)
\leq\frac{1}{8}\right\}  \cap\left\{  \frac{R_{1}}{2}\leq\epsilon_{3}\leq
\frac{3R_{1}}{2}\right\}  }\Psi\left(  4Q\left(  t_{n},\epsilon_{3}\right)
-1\right)  d\epsilon_{3}\right] \\
& \geq\sqrt{\frac{R_{1}}{2}}\Psi\left(  -\frac{1}{2}\right)  \left(
\int_{\mathcal{U}_{n}\left(  R_{1},R_{2}\right)  }f\left(  t_{n}%
,\epsilon\right)  d\epsilon\right)  ^{2}\left[  \int_{\left\{  Q\left(
t_{n},\epsilon_{3}\right)  \leq\frac{1}{8}\right\}  \cap\left\{  \frac{R_{1}%
}{2}\leq\epsilon_{3}\leq\frac{3R_{1}}{2}\right\}  }d\epsilon_{3}\right]
\end{align*}

Due to the first formula in (\ref{G4E9})\ we have that
\[
\left\vert \left\{  Q\left(  t_{n},\epsilon_{3}\right)  >\frac{1}{8}\right\}
\cap\left\{  \frac{R_{1}}{2}\leq\epsilon_{3}\leq\frac{3R_{1}}{2}\right\}
\right\vert \rightarrow0
\]
as $n\rightarrow\infty.$ Then:%
\[
\lim_{n\rightarrow\infty}\int_{\left\{  Q\left(  t_{n},\epsilon_{3}\right)
\leq\frac{1}{8}\right\}  \cap\left\{  \frac{R_{1}}{2}\leq\epsilon_{3}\leq
\frac{3R_{1}}{2}\right\}  }d\epsilon_{3}=\int_{\left\{  \frac{R_{1}}{2}%
\leq\epsilon_{3}\leq\frac{3R_{1}}{2}\right\}  }d\epsilon_{3}=R_{1}%
\]

Therefore, using (\ref{G4E5}):%
\[
\lim_{n\rightarrow\infty}\int_{\mathcal{U}_{n}\left(  R_{1},R_{2}\right)
}f\left(  t_{n},\epsilon\right)  d\epsilon=0
\]

We then have, taking into account (\ref{G4E2a}) and the first limit in
(\ref{G4E9}):
\begin{align*}
& \lim_{n\rightarrow\infty}\int_{R}^{\infty}f\left(  t_{n},\epsilon\right)
\sqrt{\epsilon}d\epsilon=\lim_{n\rightarrow\infty}\int_{\left[  R,\infty
\right)  \setminus\mathcal{U}_{n}\left(  R_{1},R_{2}\right)  }f\left(
t_{n},\epsilon\right)  \sqrt{\epsilon}d\epsilon+\\
& +\lim_{n\rightarrow\infty}\int_{\mathcal{U}_{n}\left(  R_{1},R_{2}\right)
}f\left(  t_{n},\epsilon\right)  d\epsilon\leq2\lim_{n\rightarrow\infty}%
\int_{\left[  R,\infty\right)  \setminus\mathcal{U}_{n}\left(  R_{1}%
,R_{2}\right)  }Q\left(  t_{n},\epsilon\right)  \sqrt{\epsilon}d\epsilon=0.
\end{align*}
\hfill
\end{proof}

\begin{lemma}
\label{FstarNonzero}Suppose that $f_{0},f$ are as in Lemma \ref{LemmaSubSeq}.
Let us assume that $Q_{\ast}\not \equiv 0$ for $\epsilon\geq R_{1}$ in
(\ref{G4E6}). Suppose that $0<R_{1}\leq1.$ Then:%
\begin{equation}
\int_{R_{1}}^{\infty}\int_{R_{1}}^{\infty}\int_{\frac{R_{1}}{2}}^{\frac
{3R_{1}}{2}}\left\vert Q\left(  t_{n},\epsilon_{3}\right)  Q\left(
t_{n},\epsilon_{4}\right)  -Q\left(  t_{n},\epsilon_{1}\right)  Q\left(
t_{n},\epsilon_{2}\right)  \right\vert d\epsilon_{1}d\epsilon_{2}d\epsilon
_{3}\rightarrow0\text{ as }n\rightarrow\infty\label{G5E2}%
\end{equation}

\end{lemma}

\begin{proof}
Using (\ref{G4E2}) we obtain:%
\begin{align*}
& \int_{\mathbb{R}^{+}}\int_{\mathbb{R}^{+}}\int_{\mathbb{R}^{+}}Q\left(
t_{n},\epsilon_{1}\right)  Q\left(  t_{n},\epsilon_{2}\right)  \Psi\left(
\frac{Q\left(  t_{n},\epsilon_{3}\right)  Q\left(  t_{n},\epsilon_{4}\right)
}{Q\left(  t_{n},\epsilon_{1}\right)  Q\left(  t_{n},\epsilon_{2}\right)
}-1\right)  \Phi d\epsilon_{1}d\epsilon_{2}d\epsilon_{3}\\
& \leq\int_{\mathbb{R}^{+}}\int_{\mathbb{R}^{+}}\int_{\mathbb{R}^{+}}f\left(
t_{n},\epsilon_{1}\right)  f\left(  t_{n},\epsilon_{2}\right)  \Psi\left(
\frac{Q\left(  t_{n},\epsilon_{3}\right)  Q\left(  t_{n},\epsilon_{4}\right)
}{Q\left(  t_{n},\epsilon_{1}\right)  Q\left(  t_{n},\epsilon_{2}\right)
}-1\right)  \Phi d\epsilon_{1}d\epsilon_{2}d\epsilon_{3}%
\end{align*}

Then:%
\begin{align*}
& \int_{\mathbb{R}^{+}}\int_{\mathbb{R}^{+}}\int_{\mathbb{R}^{+}}f\left(
t_{n},\epsilon_{1}\right)  f\left(  t_{n},\epsilon_{2}\right)  \Psi\left(
\frac{Q\left(  t_{n},\epsilon_{3}\right)  Q\left(  t_{n},\epsilon_{4}\right)
}{Q\left(  t_{n},\epsilon_{1}\right)  Q\left(  t_{n},\epsilon_{2}\right)
}-1\right)  \Phi d\epsilon_{1}d\epsilon_{2}d\epsilon_{3}\\
& \geq\int_{R_{1}}^{\infty}\int_{R_{1}}^{\infty}\int_{\frac{R_{1}}{2}}%
^{\frac{3R_{1}}{2}}Q\left(  t_{n},\epsilon_{1}\right)  Q\left(  t_{n}%
,\epsilon_{2}\right)  \Psi\left(  \frac{Q\left(  t_{n},\epsilon_{3}\right)
Q\left(  t_{n},\epsilon_{4}\right)  }{Q\left(  t_{n},\epsilon_{1}\right)
Q\left(  t_{n},\epsilon_{2}\right)  }-1\right)  \Phi d\epsilon_{1}%
d\epsilon_{2}d\epsilon_{3}\\
& \geq\sqrt{\frac{R_{1}}{2}}\int_{R_{1}}^{\infty}\int_{R_{1}}^{\infty}%
\int_{\frac{R_{1}}{2}}^{\frac{3R_{1}}{2}}Q\left(  t_{n},\epsilon_{1}\right)
Q\left(  t_{n},\epsilon_{2}\right)  \Psi\left(  \frac{Q\left(  t_{n}%
,\epsilon_{3}\right)  Q\left(  t_{n},\epsilon_{4}\right)  }{Q\left(
t_{n},\epsilon_{1}\right)  Q\left(  t_{n},\epsilon_{2}\right)  }-1\right)
d\epsilon_{1}d\epsilon_{2}d\epsilon_{3}%
\end{align*}

We define a convex function $\bar{\Psi}\left(  s\right)  $ as:%
\[
\bar{\Psi}\left(  s\right)  =\Psi\left(  s\right)  \ \text{if\ }s\geq0,\text{
and }\bar{\Psi}\left(  s\right)  =\bar{\Psi}\left(  -s\right)
\]

Notice that $\bar{\Psi}\left(  s\right)  $ is bounded as $C\left\vert
s\right\vert $ for large $\left\vert s\right\vert .$ Using Jensen's inequality
we obtain:%
\begin{align*}
& R_{1}\left(  \int_{R_{1}}^{\infty}Q\left(  t_{n},\epsilon\right)
d\epsilon\right)  ^{2}\sqrt{\frac{R_{1}}{2}}\bar{\Psi}\times\\
& \times\left(  \frac{\int_{R_{1}}^{\infty}\int_{R_{1}}^{\infty}\int
_{\frac{R_{1}}{2}}^{\frac{3R_{1}}{2}}\left\vert Q\left(  t_{n},\epsilon
_{3}\right)  Q\left(  t_{n},\epsilon_{4}\right)  -Q\left(  t_{n},\epsilon
_{1}\right)  Q\left(  t_{n},\epsilon_{2}\right)  \right\vert d\epsilon
_{1}d\epsilon_{2}d\epsilon_{3}}{\left(  \int_{R_{1}}^{\infty}Q\left(
t_{n},\epsilon\right)  d\epsilon\right)  ^{2}R_{1}}\right) \\
& \leq\int_{\mathbb{R}^{+}}\int_{\mathbb{R}^{+}}\int_{\mathbb{R}^{+}}f\left(
t_{n},\epsilon_{1}\right)  f\left(  t_{n},\epsilon_{2}\right)  \Psi\left(
\frac{Q\left(  t_{n},\epsilon_{3}\right)  Q\left(  t_{n},\epsilon_{4}\right)
}{Q\left(  t_{n},\epsilon_{1}\right)  Q\left(  t_{n},\epsilon_{2}\right)
}-1\right)  \Phi d\epsilon_{1}d\epsilon_{2}d\epsilon_{3}%
\end{align*}

Taking into account (\ref{G4E5}) as well as the fact that $\bar{\Psi}\left(
s\right)  $ vanishes only for $s=0$ we obtain (\ref{G5E2}).
\end{proof}

\begin{lemma}
\label{maxwellian}Given$\ 0<R_{1}<R_{2}<\infty,$ suppose that $f_{0}$ as in
the statement of Proposition \ref{ConcMass}, let $\left\{  t_{n}\right\}  $ as
in Lemma \ref{LemmaSubSeq}. Suppose that $Q_{\ast}\not \equiv 0$ in the
interval $\left[  R_{1},R_{2}\right]  .$ Then $Q_{\ast}\left(  \epsilon
\right)  =\exp\left(  -\beta_{\ast}\left(  \epsilon+\alpha\right)  \right)  $
for $\epsilon\in$ $\left[  R_{1},R_{2}\right]  $ with $\beta_{\ast}>0$ and
$\alpha\geq-R_{1}.$ Moreover, there exists a subsequence of $\left\{
t_{n}\right\}  $ which will be denoted with the same indexes such that:%
\begin{equation}
Q\left(  t_{n},\cdot\right)  \rightarrow Q_{\ast}\left(  \cdot\right)
\ \ \text{in\ }L^{1}\left(  R_{1},R_{2}\right)  \ \ \text{as\ \ }%
n\rightarrow\infty\ \ .\label{G6E4}%
\end{equation}

\end{lemma}

\begin{proof}
Let us define:%
\begin{equation}
R_{\ast}=\sup\left\{  r\in\left[  R_{1},R_{2}\right]  :\int_{R_{1}}^{r}%
Q_{\ast}d\epsilon=0\right\} \label{G5E3}%
\end{equation}

Since $Q_{\ast}\not \equiv 0$ in the interval $\left[  R_{1},R_{2}\right]  $
we have $R_{\ast}<R_{2}.$ We then define also:
\[
\mathcal{Q}_{n}\left(  \epsilon_{3}\right)  =\int_{\bar{R}}^{\infty}\int
_{\bar{R}}^{\infty}\left\vert Q\left(  t_{n},\epsilon_{3}\right)  Q\left(
t_{n},\epsilon_{1}+\epsilon_{2}-\epsilon_{3}\right)  -Q\left(  t_{n}%
,\epsilon_{1}\right)  Q\left(  t_{n},\epsilon_{2}\right)  \right\vert
d\epsilon_{1}d\epsilon_{2}%
\]
where $\bar{R}=\max\left\{  R_{1},\theta R_{\ast}\right\}  ,\ $where
$\theta<1$ is very close to one, to be determined later.

Using Lemma \ref{FstarNonzero}, and more precisely (\ref{G5E2}), with
$R_{1}=\bar{R},$ it follows that there exists a set $\mathcal{V}\subset\left[
\frac{\bar{R}}{2},\frac{3\bar{R}}{2}\right]  $ with measure $\left\vert
\mathcal{V}\right\vert =\bar{R}$ and a subsequence of $\left\{  t_{n}\right\}
,$ which will be labelled with the same indexes, such that:%
\begin{equation}
\mathcal{Q}_{n}\left(  \epsilon_{3}\right)  \rightarrow
0\ \ \,,\ \ n\rightarrow\infty\ \ \text{for any }\epsilon_{3}\in
\mathcal{V}\label{G5E8}%
\end{equation}

The definition of $R_{\ast}$ and the assumptions of the Lemma imply the
existence of a $\delta>0$ such that
\begin{equation}
\int_{R_{\ast}+\delta}^{R_{2}}Q_{\ast}\left(  \epsilon\right)  d\epsilon
>0\label{G5E8a}%
\end{equation}

We now claim that there exists $\epsilon_{\ast}\in\mathcal{V}\cap\left\{
R_{\ast}\leq\epsilon\leq R_{\ast}+\delta\right\}  $ and a subsequence of
$\left\{  t_{n}\right\}  $, labelled with the same indexes such that:%
\begin{equation}
\lim_{n\rightarrow\infty}Q\left(  t_{n},\epsilon_{\ast}\right)  =\eta
>0\label{G5E4}%
\end{equation}

Indeed, otherwise we would have $\lim_{n\rightarrow\infty}Q\left(
\epsilon,t_{n}\right)  =0$ for any $\epsilon\in\mathcal{V}\cap\left\{
R_{\ast}\leq\epsilon\leq R_{\ast}+\delta\right\}  $. Lebesgue's dominated
convergence theorem, combined with (\ref{G4E2}) would imply that $Q_{\ast
}\left(  \epsilon\right)  =0$ $a.e.\ \epsilon\in\left[  \frac{\bar{R}}%
{2},R_{\ast}+\delta\right]  ,$ but this would contradict the definition of
$R_{\ast}$ in (\ref{G5E3}).

Therefore:%
\[
\int_{\bar{R}}^{\infty}\int_{\bar{R}}^{\infty}\left\vert Q\left(
t_{n},\epsilon_{\ast}\right)  Q\left(  t_{n},\epsilon_{1}+\epsilon
_{2}-\epsilon_{\ast}\right)  -Q\left(  t_{n},\epsilon_{1}\right)  Q\left(
t_{n},\epsilon_{2}\right)  \right\vert d\epsilon_{1}d\epsilon_{2}%
\rightarrow0\ \ \text{as\ \ }n\rightarrow\infty
\]

Writing $x=\epsilon_{1}-\epsilon_{\ast},$ $y=\epsilon_{2}-\epsilon_{\ast}$ and
$H_{n}\left(  x\right)  =Q\left(  t_{n},\epsilon_{\ast}+x\right)  $ it then
follows that:%
\begin{equation}
\int_{\bar{R}-\epsilon_{\ast}}^{\infty}\int_{\bar{R}-\epsilon_{\ast}}^{\infty
}\left\vert Q\left(  t_{n},\epsilon_{\ast}\right)  H_{n}\left(  x+y\right)
-H_{n}\left(  x\right)  H_{n}\left(  y\right)  \right\vert dxdy\rightarrow
0\ \ \text{as\ \ }n\rightarrow\infty\label{G5E5}%
\end{equation}

We define:%
\begin{equation}
G_{n}\left(  x\right)  =\int_{x}^{\infty}H_{n}\left(  \xi\right)
d\xi\label{G5E6}%
\end{equation}

Note that this integral is finite for any $n,$ due to (\ref{G4E2}). Moreover,
they are uniformly bounded by a constant $C=C\left(  M,R_{1}\right)  <\infty.$
Then:\
\begin{align*}
& \int_{\bar{R}-\epsilon_{\ast}}^{\infty}\left\vert Q\left(  t_{n}%
,\epsilon_{\ast}\right)  G_{n}\left(  x\right)  -H_{n}\left(  x\right)
G_{n}\left(  0\right)  \right\vert dx\\
& \leq\int_{\bar{R}-\epsilon_{\ast}}^{\infty}\int_{0}^{\infty}\left\vert
Q\left(  t_{n},\epsilon_{\ast}\right)  H_{n}\left(  x+y\right)  -H_{n}\left(
x\right)  H_{n}\left(  y\right)  \right\vert dydx\\
& \leq\int_{\bar{R}-\epsilon_{\ast}}^{\infty}\int_{\bar{R}-\epsilon_{\ast}%
}^{\infty}\left\vert Q\left(  t_{n},\epsilon_{\ast}\right)  H_{n}\left(
x+y\right)  -H_{n}\left(  x\right)  H_{n}\left(  y\right)  \right\vert dy
dx\rightarrow0
\end{align*}
where we use that $\epsilon_{\ast}\geq R_{\ast}\geq\bar{R}.$

Defining:%
\begin{equation}
\tilde{\lambda}_{n}\left(  x\right)  =Q\left(  t_{n},\epsilon_{\ast}\right)
G_{n}\left(  x\right)  -H_{n}\left(  x\right)  G_{n}\left(  0\right)
\ \label{G5E7}%
\end{equation}
it then follows that:%
\begin{equation}
\int_{\bar{R}-\epsilon_{\ast}}^{\infty}\left\vert \tilde{\lambda}_{n}\left(
x\right)  \right\vert dx\rightarrow0\ \ \text{as\ \ }n\rightarrow
\infty\label{G5E7a}%
\end{equation}

Due to (\ref{G4E2}) we have $G_{n}\in W^{1,\infty}\left(  R_{\ast}%
-\epsilon_{\ast},\infty\right)  .$ Then (\ref{G5E7}) implies:%
\begin{equation}
\tilde{\lambda}_{n}\left(  x\right)  =Q\left(  t_{n},\epsilon_{\ast}\right)
G_{n}\left(  x\right)  +G_{n}\left(  0\right)  G_{n}^{\prime}\left(  x\right)
\ \ ,\ \ a.e.\ x\geq\bar{R}-\epsilon_{\ast}\label{G5E9}%
\end{equation}
Notice also that the weak convergence of the sequence $Q\left(  \cdot
,t_{n}\right)  $ implies that $G_{n}\left(  0\right)  \rightarrow
\int_{\epsilon_{\ast}}^{\infty}Q_{\ast}\left(  \epsilon\right)  d\epsilon$ as
$n\rightarrow\infty.$ Due to (\ref{G5E8a}) we have $\int_{\epsilon_{\ast}%
}^{\infty}Q_{\ast}\left(  \epsilon\right)  d\epsilon>0.$ Let us write
$\lambda_{n}\left(  x\right)  =\frac{\tilde{\lambda}_{n}\left(  x\right)
}{G_{n}\left(  0\right)  },$ $\beta_{n}=\frac{Q\left(  t_{n},\epsilon_{\ast
}\right)  }{G_{n}\left(  0\right)  }.$ Due to (\ref{G5E4}), (\ref{G5E7a}) and
the fact that $\lim_{n\rightarrow\infty}G_{n}\left(  0\right)  >0$, we obtain:%
\begin{equation}
\int_{\bar{R}-\epsilon_{\ast}}^{\infty}\left\vert \lambda_{n}\left(  x\right)
\right\vert dx\rightarrow0\ \ ,\ \ \beta_{n}\rightarrow\beta_{\ast
}>0\ \ \text{as\ \ }n\rightarrow\infty\ \ \ \label{G6E1}%
\end{equation}%
\begin{equation}
G_{n}^{\prime}\left(  x\right)  +\beta_{n}G_{n}\left(  x\right)  =\lambda
_{n}\left(  x\right)  \ \label{G6E2}%
\end{equation}

Integrating (\ref{G6E2}):%
\begin{equation}
G_{n}\left(  x\right)  =A_{n}e^{-\beta_{n}x}+\int_{x}^{\infty}e^{-\beta
_{n}\left(  x-y\right)  }\lambda_{n}\left(  y\right)  dy\ \ ,\ \ x\geq\bar
{R}-\epsilon_{\ast}\ \label{G6E3}%
\end{equation}
for suitable constants $A_{n}\in\mathbb{R}$. Due to (\ref{G6E1}) the integral
term in (\ref{G6E3}) converges to zero, uniformly in the set $\left[
0,\infty\right)  .$ Since $G_{n}\left(  0\right)  $ is uniformly bounded it
then follows that the sequence $\left\{  A_{n}\right\}  $ is bounded. Taking a
new subsequence if needed, it then follows that $A_{n}\rightarrow A_{\ast}$ as
$n\rightarrow\infty,$ whence:%
\[
G_{n}\left(  x\right)  \rightarrow A_{\ast}e^{-\beta_{\ast}x}\ \ ,\ \ x\geq
\bar{R}-\epsilon_{\ast}%
\]

Using (\ref{G6E2}) we obtain that $H_{n}\rightarrow A_{\ast}\beta_{\ast
}e^{-\beta_{\ast}x}$ in $L^{1}\left(  \bar{R}-\epsilon_{\ast},\bar{L}\right)
$ for any $\bar{L}$ fixed, sufficiently large. Therefore
\begin{equation}
Q\left(  t_{n},\epsilon\right)  \rightarrow\exp\left(  -\beta_{\ast}\left(
\epsilon+\alpha\right)  \right)  \ \ \text{as\ \ }n\rightarrow\infty
\ \ ,\ \ \alpha\in\mathbb{R\ }\label{G6E5}%
\end{equation}
in $L^{1}\left(  \bar{R},L\right)  $ with $L$ large. We now consider two
cases. If $R_{\ast}=R_{1},$ (\ref{G6E5}) would imply (\ref{G6E4}). Otherwise,
we choose $\theta<1$ sufficiently close to one to have $R_{1}\leq\theta
R_{\ast}<R_{\ast}.$ Then (\ref{G6E5}) would imply $Q_{\ast}\left(
\epsilon\right)  =\exp\left(  -\beta_{\ast}\left(  \epsilon+\alpha\right)
\right)  $ and this would contradict the definition of $R_{\ast}$ in
(\ref{G5E3}). Therefore $\bar{R}=R_{\ast}=R_{1}.$ Using that $Q_{\ast}\left(
\epsilon\right)  \leq1$ for $\epsilon\geq R_{1},$ it follows that $\alpha
\geq-R_{1}$ and this concludes the Proof.
\end{proof}

\begin{lemma}
\label{LemmL1}Suppose that $f_{0}$ as in the statement of Proposition
\ref{ConcMass}. Let $0<R_{1}\leq1<R_{2}<\infty$ and $\left\{  t_{n}\right\}  $
as in Lemma \ref{maxwellian}. Let us write $f_{\ast}\left(  \epsilon\right)
=\frac{1}{\exp\left(  \beta_{\ast}\left(  \epsilon+\alpha\right)  \right)
-1}$ where $\alpha,\ \beta_{\ast}$ are as in Lemma \ref{maxwellian}. Then, for
a suitable subsequence of $\left\{  t_{n}\right\}  $ which will be labelled
with the same indexes and for any $\delta>0$ small:%
\[
f\left(  t_{n},\cdot\right)  \rightarrow f_{\ast}\left(  \cdot\right)
\ \ \text{in\ }L^{1}\left(  \left(  R_{1},R_{2}\right)  \cap\left\{  Q_{\ast
}<1-5\delta\right\}  \right)  \ \ \text{as\ \ }n\rightarrow\infty
\]
with $Q_{\ast}$ as in Lemma \ref{maxwellian}.
\end{lemma}

\begin{remark}
Notice that the restriction $\left\{  Q_{\ast}<1-5\delta\right\}  $ applies
only to the limit distribution $Q_{\ast},$ and not to the sequence $Q\left(
t_{n},\cdot\right)  $.
\end{remark}

\begin{proof}
Taking a subsequence if needed we can assume that the convergence $Q\left(
t_{n},\cdot\right)  \rightarrow Q_{\ast}\left(  \cdot\right)  $ in Lemma
\ref{maxwellian} takes place for $a.e.\ \epsilon\in\left(  R_{1},R_{2}\right)
.$ Let us denote as $\mathcal{I}$ the set \ $\mathcal{I}=\left(  R_{1}%
,R_{2}\right)  \cap\left\{  \epsilon\, ;\, Q_{\ast}\left(  \epsilon\right)
<1-5\delta\right\}  .$ We estimate the $L^{1}$ norm of $f\left(  t_{n}%
,\cdot\right)  -f_{\ast}\left(  \cdot\right)  $ as follows:%
\[
\int_{\mathcal{I}}\left\vert f\left(  t_{n},\epsilon\right)  -f_{\ast}\left(
\epsilon\right)  \right\vert d\epsilon=J_{1,n}+J_{2,n}%
\]%
\begin{align}
& J_{1,n}=\int_{\mathcal{I}\cap\mathcal{B}_{n}}\left\vert f\left(
t_{n},\epsilon\right)  -f_{\ast}\left(  \epsilon\right)  \right\vert
d\epsilon\label{G7E2a}\\
& J_{2,n}=\int_{\mathcal{I}\cap\mathcal{B}_{n}^{c} }\left\vert f\left(
t_{n},\epsilon\right)  -f_{\ast}\left(  \epsilon\right)  \right\vert
d\epsilon\label{G7E2b}\\
& \mathcal{B}_{n}= \left\{  \epsilon>0\, ;\, Q\left(  t_{n},\epsilon\right)
\geq1-\delta\right\}
\end{align}

The sequence $J_{2,n}$ converges to zero as $n\rightarrow\infty$ due to
(\ref{G4E2a}) as well as the fact that the function $\frac{1}{\left(
1-Q\left(  t_{n},\epsilon\right)  \right)  \left(  1-Q_{\ast}\left(
\epsilon\right)  \right)  }$ is bounded in the corresponding integration region.

To estimate $J_{1,n}$ we use Lemma \ref{LemmaSubSeq}. Then:%
\[
\int_{\mathcal{I}\cap\mathcal{B}_{n} }\int_{\mathcal{I}\cap\mathcal{B}_{n}
}\int_{\mathcal{I}}f\left(  t_{n},\epsilon_{1}\right)  f\left(  t_{n}%
,\epsilon_{2}\right)  \Psi\left(  \frac{Q\left(  t_{n},\epsilon_{3}\right)
Q\left(  t_{n},\epsilon_{4}\right)  }{Q\left(  t_{n},\epsilon_{1}\right)
Q\left(  t_{n},\epsilon_{2}\right)  }-1\right)  \Phi d\epsilon_{3}
d\epsilon_{1}d\epsilon_{2}\rightarrow0\ \ \text{as\ \ }n\rightarrow\infty\
\]

Notice that $\Phi$ can be estimated from below uniformly in $n$ if, say,
\begin{equation}
R_{1}\leq\epsilon_{3}\leq\frac{2}{3}\left(  \epsilon_{1}+\epsilon_{2}\right)
\ \label{G7E1}%
\end{equation}

On the other hand, since we integrate in $\epsilon_{3}$ in $\mathcal{I}$ we
need to ensure that the domain where (\ref{G7E1}) holds has an intersection
with $\mathcal{I}$ whose measure can be estimated from below. This can be seen
because the values of $\epsilon_{1},\ \epsilon_{2}$ must be also in the
interval $\mathcal{I}$. Notice that $\mathcal{I}=\left(  \bar{\epsilon}%
,R_{2}\right)  $ for some $\bar{\epsilon}$ depending on $\delta,R_{1}.$
Therefore we obtain that the region of integration for $\epsilon_{3}$ can be
replaced by a set $\left(  \bar{\epsilon},\frac{4\bar{\epsilon}}{3}\right)  .$
Notice that this set is contained in $\mathcal{I}$. Then:%
\[
\int_{\mathcal{I}\cap\mathcal{B}_{n} }\int_{\mathcal{I}\cap\mathcal{B}_{n}
}\int_{\left(  \bar{\epsilon},\frac{4\bar{\epsilon}}{3}\right)  }f\left(
t_{n},\epsilon_{1}\right)  f\left(  t_{n},\epsilon_{2}\right)  \Psi\left(
\frac{Q\left(  t_{n},\epsilon_{3}\right)  Q\left(  t_{n},\epsilon_{4}\right)
}{Q\left(  t_{n},\epsilon_{1}\right)  Q\left(  t_{n},\epsilon_{2}\right)
}-1\right)  d\epsilon_{3} d\epsilon_{1}d\epsilon_{2}\rightarrow
0\ \ \text{as\ \ }n\rightarrow\infty
\]
due to the fact that we have a lower estimate for $\Phi$ independent on $n$
for $\epsilon_{3}\in\left(  \bar{\epsilon},\frac{4\bar{\epsilon}}{3}\right)
.$ We now use the convergence of $Q\left(  t_{n},\cdot\right)  $ to $Q_{\ast
}\left(  \cdot\right)  $ in $L^{1}\left(  \bar{\epsilon},\frac{4\bar{\epsilon
}}{3}\right)  $ which is a Corollary of Lemma \ref{maxwellian}. Egoroff's
Theorem shows that there exists a set $\mathcal{A}\subset\left(  \bar
{\epsilon},\frac{4\bar{\epsilon}}{3}\right)  $ with a measure arbitrarily
close to $\frac{\bar{\epsilon}}{3}$ where $Q\left(  t_{n},\cdot\right)
\rightarrow Q_{\ast}\left(  \cdot\right)  $ uniformly. Then:%
\[
\int_{\mathcal{I}\cap\mathcal{B}_{n} }\int_{\mathcal{I}\cap\mathcal{B}_{n}%
}\int_{\mathcal{A}}f\left(  t_{n},\epsilon_{1}\right)  f\left(  t_{n}%
,\epsilon_{2}\right)  \Psi\left(  \frac{Q\left(  t_{n},\epsilon_{3}\right)
Q\left(  t_{n},\epsilon_{4}\right)  }{Q\left(  t_{n},\epsilon_{1}\right)
Q\left(  t_{n},\epsilon_{2}\right)  }-1\right)  d\epsilon
_{1}d\epsilon_{2}d \epsilon_{3}\rightarrow0\ \ \text{as\ \ }n\rightarrow\infty
\]
and assuming that $n$ is sufficiently large we would have:%
\[
\frac{Q\left(  t_{n},\epsilon_{3}\right)  Q\left(  t_{n},\epsilon_{4}\right)
}{Q\left(  t_{n},\epsilon_{1}\right)  Q\left(  t_{n},\epsilon_{2}\right)
}\leq\frac{\left(  1-5\delta\right)  }{\left(  1-\delta\right)  ^{2}}%
\leq\left(  1-\delta\right)
\]
if $\delta$ is small. Then, $\Psi\left(  \frac{Q\left(  t_{n},\epsilon
_{3}\right)  Q\left(  t_{n},\epsilon_{4}\right)  }{Q\left(  t_{n},\epsilon
_{1}\right)  Q\left(  t_{n},\epsilon_{2}\right)  }-1\right)  \geq c_{0}>0$
independent of $n,$ whence:%
\[
c_{0}\left\vert \mathcal{A}\right\vert \left(  \int_{\mathcal{I}%
\cap\mathcal{B}_{n} }f\left(  t_{n},\epsilon\right)  d\epsilon\right)
^{2}\rightarrow0\ \ \text{as\ \ }n\rightarrow\infty
\]

Then:%
\begin{equation}
\int_{\mathcal{I}\cap\mathcal{B}_{n} }f\left(  t_{n},\epsilon\right)
d\epsilon\rightarrow0\ \ \text{as\ \ }n\rightarrow\infty\label{G7E3}%
\end{equation}

On the other hand, since $Q_{\ast}\left(  \epsilon\right)  \leq\left(
1-5\delta\right)  $ for $\epsilon\in\mathcal{I}$, it follows that:%
\begin{equation}
\left\vert \mathcal{I}\cap\mathcal{B}_{n} \right\vert \rightarrow0\ \ \text{as
}n\rightarrow\infty\ \label{G7E4}%
\end{equation}
due to Lebesgue's Theorem. We can take the limit of the sequence $J_{1,n}$ as:%
\[
J_{1,n}\leq\int_{\mathcal{I}\cap\mathcal{B}_{n} }f\left(  t_{n},\epsilon
\right)  d\epsilon+\int_{\mathcal{I}\cap\mathcal{B}_{n} }f_{\ast}\left(
\epsilon\right)  d\epsilon\rightarrow0
\]
due to the fact that $f_{\ast}\left(  \epsilon\right)  $ is bounded in
$\mathcal{I}$ and using (\ref{G7E3}), (\ref{G7E4}).
\end{proof}

\begin{proof}
[Proof of Proposition \ref{ConcMass}]Given $f_{0}$ as in the statement of
Proposition \ref{ConcMass}, we have that the stationary solution
$F_{BE}\left(  p;\alpha,\beta,m_{0}\right)  $ in Proposition \ref{crit} having
the number of particles $M$ and energy $E$ satisfies $m_{0}>0.$ We define
$m_{\ast}=\frac{m_{0}}{2}.$ Due to the Lemma \ref{LemmaL} we can select $L>0$
such that:%
\begin{equation}
\int_{L}^{\infty}g\left(  t,\epsilon\right)  d\epsilon=4\pi\int_{L}^{\infty
}f\left(  t,\epsilon\right)  \sqrt{2\epsilon}d\epsilon\leq\frac{m_{0}}%
{10}\ \label{G5E1}%
\end{equation}
for any $t\in\left[  0,T_{\max}\right]  .$

We now apply Lemma \ref{LemmaSubSeq}. Then (\ref{G4E6}) holds. We have now two
possibilities. Suppose first that $Q_{\ast}\equiv0$ in $\left[  \frac{R}%
{2},\infty\right)  $ $.$ Then Lemma \ref{R1R2} with $R_{1}=\frac{R}{2}$ and
$R_{2}=L,$ combined with (\ref{G5E1}) imply that:%
\[
\int_{R}^{\infty}g\left(  t_{n},\epsilon\right)  d\epsilon=\int_{R}%
^{L}g\left(  t_{n},\epsilon\right)  d\epsilon+\int_{L}^{\infty}g\left(
t_{n},\epsilon\right)  d\epsilon\leq\int_{\frac{R}{2}}^{L}g\left(
t_{n},\epsilon\right)  d\epsilon+\frac{m_{0}}{10}\leq\frac{m_{0}}{5}%
\]
if $n$ is sufficiently large. Then:%
\[
\int_{0}^{R}g\left(  t_{n},\epsilon\right)  d\epsilon\geq\frac{4m_{0}}%
{5}>m_{\ast}%
\]
whence the result follows in this case.

Suppose now that $Q_{\ast}\not \equiv 0$ in $\left[  \frac{R}{2}%
,\infty\right)  .$ We then set $R_{1}=\frac{R}{2},$ $R_{2}=L.$ We choose
$\delta>0$ small and apply Lemma \ref{LemmL1} to obtain:%
\begin{equation}
\int_{\left(  R,L\right)  \cap\left\{  Q_{\ast}<1-5\delta\right\}  }g\left(
t_{n},\epsilon\right)  d\epsilon\rightarrow4\pi\int_{\left(  R,L\right)
\cap\left\{  Q_{\ast}<1-5\delta\right\}  }f_{\ast}\left(  \epsilon\right)
\sqrt{2\epsilon}d\epsilon\ \label{G7E5}%
\end{equation}%
\[
\int_{\left(  R,L\right)  \cap\left\{  Q_{\ast}<1-5\delta\right\}  }g\left(
t_{n},\epsilon\right)  \epsilon d\epsilon\rightarrow4\pi\int_{\left(
R,L\right)  \cap\left\{  Q_{\ast}<1-5\delta\right\}  }f_{\ast}\left(
\epsilon\right)  \sqrt{2\epsilon^{3}}d\epsilon
\]
as $n\rightarrow\infty.$ (Notice that the integrations are made in the
interval $\left(  R,L\right)  $ in spite of the fact that the convergence in
Lemma \ref{maxwellian} is in $\epsilon\geq\frac{R}{2}.$ In particular
$\alpha\geq-\frac{R}{2}$).

We now claim that the right-hand side of (\ref{G7E5}) can be made smaller than
$M$ plus some small error term if $\delta$ and $R$ are small. Actually taking
$\delta$ much smaller than $\frac{R}{2}$ we would obtain that $\left(
R,L\right)  \cap\left\{  Q_{\ast}<1-5\delta\right\}  =\left(  R,L\right)  .$
Then:%
\begin{equation}
\int_{\left(  R,L\right)  }g\left(  t_{n},\epsilon\right)  d\epsilon
\rightarrow4\pi\int_{\left(  R,L\right)  }f_{\ast}\left(  \epsilon\right)
\sqrt{2\epsilon}d\epsilon\label{G9E1}%
\end{equation}%
\[
\int_{\left(  R,L\right)  }g\left(  t_{n},\epsilon\right)  \epsilon
d\epsilon\rightarrow4\pi\int_{\left(  R,L\right)  }f_{\ast}\left(
\epsilon\right)  \sqrt{2\epsilon^{3}}d\epsilon
\]

We now use that $\int_{\left(  R,L\right)  }g\left(  t_{n},\epsilon\right)
\epsilon d\epsilon\leq E,$ whence:%
\[
4\pi\int_{\left(  R,L\right)  }f_{\ast}\left(  \epsilon\right)  \sqrt
{2\epsilon^{3}}d\epsilon\leq E
\]

On the other hand, using (\ref{G5E1}) as well as the fact that $\int
gd\epsilon=M$ we obtain that:%
\[
\int_{\left(  0,R\right)  }g\left(  t_{n},\epsilon\right)  d\epsilon
+\int_{\left(  R,L\right)  }g\left(  t_{n},\epsilon\right)  d\epsilon\geq
M-\frac{m_{0}}{10}%
\]

Using (\ref{G9E1}) as well as Proposition \ref{MassCrit} we obtain, assuming
that $n$ is large and taking $\delta=\frac{m_{0}}{10}$ in Proposition
\ref{MassCrit}:%
\[
\int_{\left(  R,L\right)  }g\left(  t_{n},\epsilon\right)  d\epsilon\leq
\frac{\zeta\left(  \frac{3}{2}\right)  }{\left(  \zeta\left(  \frac{5}%
{2}\right)  \right)  ^{\frac{3}{5}}}\left(  \frac{4\pi}{3}\right)  ^{\frac
{3}{5}}E^{\frac{3}{5}}+\frac{m_{0}}{10}%
\]

Using now the hypothesis (\ref{G1E2}) in Theorem \ref{Main} as well as the
fact that $M-\frac{\zeta\left(  \frac{3}{2}\right)  }{\left(  \zeta\left(
\frac{5}{2}\right)  \right)  ^{\frac{3}{5}}}\left(  \frac{4\pi}{3}\right)
^{\frac{3}{5}}E^{\frac{3}{5}}=m_{0}$ we obtain:%
\begin{equation}
\int_{\left(  0,R\right)  }g\left(  t_{n},\epsilon\right)  d\epsilon
\geq\left[  M-\frac{\zeta\left(  \frac{3}{2}\right)  }{\left(  \zeta\left(
\frac{5}{2}\right)  \right)  ^{\frac{3}{5}}}\left(  \frac{4\pi}{3}\right)
^{\frac{3}{5}}E^{\frac{3}{5}}\right]  -\frac{m_{0}}{5}=\frac{4m_{0}}%
{5}\label{G9E2}%
\end{equation}
if $R$ is sufficiently small, depending only on $E,\ M.$ We now use
Proposition \ref{lowMass} combined with (\ref{G9E2}) and Theorem \ref{main}
implies that $f$ must become unbounded in finite time.\ Notice that we apply
Theorem \ref{main} taking as starting time $T_{1}=t_{n}.$ The Theorem
\ref{main} can be applied there due to the invariant of the problem under
translations in time. Note also that given a mild solution of (\ref{F3E2}),
(\ref{F3E3}) defines also a mild solution in any time interval $\left[
T_{1},T_{2}\right]  \subset\left[  0,T_{\max}\right)  .$ Indeed, we just need
to split the integral in time and using the remaining terms as new initial
data for the solution using the semigroup property of the exponential term.
\end{proof}

\section{End of the Proof of Theorem \ref{Main}}

\setcounter{equation}{0} \setcounter{theorem}{0}

\begin{proof}
[End of the Proof of Theorem \ref{Main}]It is just a consequence of Theorem
\ref{main}, Proposition \ref{lowMass} and Proposition \ref{ConcMass}. Notice
that the blow-up condition (\ref{main}) can be applied starting at any time
$T\geq0$ and not necessarily at $T=0,$ due to the invariance of
(\ref{F3E2}) under translations in time.
\end{proof}

\section{Finite time condensation.}

\begin{proof}[Proof of Theorem \protect\ref{Cond1}] Let  $g_{0}\in\mathcal{M}_{+}\left(  \mathbb{R}^{+};\left(  1+\epsilon\right)\right)  $ be the weak solution of  (\ref{F3E4}), (\ref{F3E5}) with initial data $g_0$, whose existence is proved in  \cite{Lu2}.  It is also proved in the same article that, when $t\to +\infty$,   $g(t)$ converges  in the  sense of measures to the unique steady state $F _{ BE }$ of the family (\ref{St1}) with the same mass $M$ and energy $E$.  Since $M$ and $E$ satisfy (\ref{G1E2}), that steady state contains a Dirac measure at the origin. We deduce that there exists  $t_n\to +\infty$,  $\rho_n \to 0$ and a constant $m>0$ such that for all $n$:
$$
\int_{0}^{\rho_n}g (t_n, \epsilon) \left( \epsilon \right)
d\epsilon\geq m. $$ 
It then follows that $g(t_n , \cdot)$ satisfies:
\begin{equation}
\int _0^\rho g(t_n, \epsilon) d\epsilon\ge K^*\rho_n ^{\theta_*}
\end{equation}
if $n$ is sufficiently large. On the other hand, by Proposition \ref{lowMass}, if $t_n>T_1(E, M)$, $g(t_n)$ also satisfies 
\begin{equation}
\int_{0}^{R}g\left(t_n, \epsilon\right)  d\epsilon\geq KR^{\frac{3}{2}%
}\ \ \text{for any }0<R\leq\rho_{1}(E, M), 
\end{equation}
and then, choosing $\nu < K$ we have:
\begin{equation}
\int _0^R g_0(\epsilon) d\epsilon\ge \nu R^{3/2}\ \ \text{for any }0<R\leq\rho_{1}(E, M).
\end{equation}
Using that $t_n\to \infty$ as $n\to \infty$, we deduce that for if $n$ is sufficiently large, $g(t_n)$ satisfies condition (\ref{Z1E6}). Therefore, since for all $t>t_n $, the function $g$ may be seen as a weak solution of (\ref{F3E4}), (\ref{F3E5}) with initial data $g(t_n )$ 
we deduce from Theorem \ref{Cond1} that there exists $T_0\ge t_n $ such that property (\ref{Z1E7}) holds. 
\end{proof}

\section{Blow-up and Condensation for subcritical data.}

\setcounter{equation}{0} \setcounter{theorem}{0}

It is worth to notice that the blow-up and condensation conditions ((\ref{Condicion}) and 
(\ref{Z1E6}) respectively) in \cite{EV1} are
 purely
local. As a consequence it is possible to find initial data with values of the
particle density and the energy $\left(  M,E\right)  $ in the subcritical
region, but on the other hand yielding blow-up in finite time or condensation. We formulate
the results as the two following Theorems due to their independent interest. 

\begin{theorem}
\label{BlowupLocal} Given $f_{0}\in L^{\infty}\left(  \mathbb{R}^{+};\left(  1+\epsilon\right)
^{\gamma}\right)  $ with $\gamma>3.$ Let us denote as $M,\ E$ the numbers:%
\[
4\pi\int_{0}^{\infty}f_{0}\left(  \epsilon\right)  \sqrt{2\epsilon}%
d\epsilon=M\ \ ,\ \ 4\pi\int_{0}^{\infty}f_{0}\left(  \epsilon\right)
\sqrt{2\epsilon^{3}}d\epsilon=E
\]
Let us denote as $f\in L_{loc}^{\infty}\left(  \left[  0,T_{\max}\right)
;L^{\infty}\left(  \mathbb{R}^{+};\left(  1+\epsilon\right)  ^{\gamma}\right)
\right)  $ the mild solution of (\ref{F3E2}), (\ref{F3E3}) in Theorem
\ref{localExistence} where $T_{\max}$ is the maximal existence time. There
exist funtions $f_{0}\in L^{\infty}\left(  \mathbb{R}^{+};\left(
1+\epsilon\right)  ^{\gamma}\right)  $ with $\gamma>3$ such that:%
\begin{equation}
M<\frac{\zeta\left(  \frac{3}{2}\right)  }{\left(  \zeta\left(  \frac{5}%
{2}\right)  \right)  ^{\frac{3}{5}}}\left(  \frac{4\pi}{3}\right)  ^{\frac
{3}{5}}E^{\frac{3}{5}}\ \label{G9E5}%
\end{equation}
for which $T_{\max}<\infty.$
\end{theorem}

\begin{proof}
It is just a consequence of the fact that (\ref{C1}), (\ref{Condicion}) can be
obtained for initial data satisfying (\ref{G9E5}). This can be done  modifying
locally $f_{0}$ in a region $\epsilon\leq R$, with $R$ small  in order to obtain   (\ref{C1}), (\ref{Condicion}).
This can be made, as in the Proof of Proposition (2.4) in \cite{EV1}, changing $M$ and $E$ as little as wished.
\end{proof}

\begin{theorem}
\label{CondLocal}
There exists weak solutions $g$ of (\ref{F3E4}), (\ref{F3E5})  in the sense of Definition \ref{weak} with initial data
$g_{0}\in\mathcal{%
M}_{+}\left( \mathbb{R}^{+};\left( 1+\epsilon\right) \right) $ satisfying$\ $%
\begin{eqnarray*}
&&4\pi\sqrt{2}\int_{\mathbb{R}^{+}}g_{0}\left( \epsilon\right) d\epsilon  =M\
,\ \ 4\pi\sqrt{2}\int_{\mathbb{R}^{+}}g_{0}\left( \epsilon\right) \epsilon
d\epsilon=E,
\end{eqnarray*}
with $M$ and $E$ like in  (\ref{G9E5}),  such that  (\ref{Z1E7})  holds for some finite $T>0$.
\end{theorem}

\begin{proof}
The Proof of Theorem  \ref{CondLocal} is completely similar to that of Theorem  \ref{BlowupLocal} with  condition  (\ref{Z1E6})  instead of  (\ref{C1}), (\ref{Condicion})
and using $g_0$ instead of $f_0$.
\end{proof}

\begin{acknowledgement}
This work has been supported by DGES Grant 2011-29306-C02-00, IT-305-07 and
the Hausdorff Center for Mathematics of the University of Bonn.
\end{acknowledgement}

\end{document}